\documentclass[11pt,a4paper]{amsart}
\usepackage{fullpage}  

\usepackage{amsmath}
\usepackage{amssymb}
\usepackage{amsthm} 
\usepackage{enumerate}
\usepackage[shortlabels]{enumitem}

\usepackage{filecontents}  
\usepackage{array}
\usepackage{longtable}

\usepackage{algorithm}
\usepackage{algpseudocode}

\usepackage[bookmarks=true,hypertexnames=false,pagebackref]{hyperref}
\hypersetup{colorlinks=true, citecolor=blue, linkcolor=blue, urlcolor=blue} 
 
\usepackage[UKenglish]{babel}
\usepackage[UKenglish]{isodate}
\usepackage[backgroundcolor=green]{todonotes}

\let\epsilon=\varepsilon

\def\*#1{\mathbf{#1}}
\def\+#1{\mathcal{#1}}
\def\-#1{\mathrm{#1}}
\def\=#1{\mathbb{#1}}

\def\Pr{\mathop{\rm Pr}\nolimits}
\newcommand{\Ex}{\mathop{\mathbb{{}E}}\nolimits}

\newcommand{\dist}{\text{\rm dist}}
\newcommand{\abs}[1]{\left\vert#1\right\vert}  
\newcommand{\aabs}[1]{|#1|}  
\newcommand{\ceil}[1]{\left\lceil #1 \right\rceil}

\newcommand{\tuple}[1]{\left(#1\right)} \newcommand{\eps}{\varepsilon}
\def\whp{w.h.p.~}
\def\Whp{W.h.p.~}
\def\true{\mathsf{T}}
\def\false{\mathsf{F}}
\def\calA{\mathcal{A}}

\def\calC{\mathcal{C}}
\def\calD{\mathcal{D}}
\def\calL{\mathcal{L}}
\def\calT{\mathcal{T}}
\def\Aone{{\calA_1}}
\def\Atwo{{\calA_2}}
\def\Ai{{\calA_i}}

\def\VI{V_I}
\def\Vset{V_{\text{set}}}
\def\Cl{\calC}
\def\Va{\+V}
\def\CLI{\Cl_I}
\def\CLO{\Cl_O}
\def\Cother{\Cl_{\text{other}}}

\def\Crem{\calC_\text{rem}}
\def\var{\mathsf{var}}

\def\marked{\mathsf{marked}}

\def\rhonext{\rho_{\tau_1,\tau_2}}
\def\rlower{r_{\mathsf{lower}}}
\def\rupper{r_{\mathsf{upper}}}
\def\plower{p_{\mathsf{lower}}}
\def\pupper{p_{\mathsf{upper}}}

\def\sigtrue{u \rightarrow \true}
\def\sigfalse{u\rightarrow \false}
\def\sigX{u \rightarrow X}
\def\sigY{u \rightarrow Y}

\def\sigXbar{u \rightarrow \neg X}
\def\failedclauses{\+F}
\def\reason{R}
\def\disagree{\mathsf{disagree}}
\def\badsf{\mathsf{bad}}
\def\rhois{\rho=(\Aone,\Atwo,\VI,\Vset,\Crem,\failedclauses,\reason)}

\def\bad{\text{bad}}
\def\good{\text{good}}
\def\HD{\mathsf{HD}}
\def\HighD{\Delta}
\def\badcomp{\mathsf{BC}}

\newcommand{\cl}{{\+L}}

\def\complement#1{\overline{#1}}
\def\partialVal{\Lambda}
\def\pseudoProb{Q}
\def\Poi{\mathrm{Poi}}
\def\choosebad{\textsf{P}}
\def\eps{\varepsilon}

\def\depthceiling{\lceil\log(n/\eps)\rceil}

\def\Lovasz{Lov\'asz}
\def\zzeta{W}

\def\ConZero{\textbf{Constraint Set 0}}
\def\ConOne{\textbf{Constraint Set 1}}
\def\ConTwo{\textbf{Constraint Set 2}}
\def\ConThree{\textbf{Constraint Set 3}}

\renewcommand{\exp}[1]{\mathrm{exp}\tuple{#1}}
\renewcommand{\deg}{\mathsf{deg}}

\newtheorem{theorem}{Theorem} 
\newtheorem{lemma}[theorem]{Lemma}
\newtheorem{corollary}[theorem]{Corollary}
\newtheorem{observation}[theorem]{Observation}  
\newtheorem{definition}[theorem]{Definition}  
  
\newtheorem*{remark*}{Remark}

\makeatletter
\def\prob#1#2#3{\goodbreak\begin{list}{}{\labelwidth\z@ \itemindent-\leftmargin
  \itemsep\z@  \topsep6\p@\@plus6\p@
  \let\makelabel\descriptionlabel}
  \item[\textbf{Name}]#1
  \item[\textbf{Instance}]#2
  \item[\textbf{Output}]#3
  \end{list}}
\makeatother

\title{Counting solutions to random CNF formulas}

\author{Andreas Galanis}
\author{Leslie Ann Goldberg}
\author{Heng Guo}
\author{Kuan Yang}
\thanks{A preliminary short version of the manuscript (without proofs) appeared in the proceedings of ICALP~2020. 
The research leading to these results has received funding from the European Research Council under
the European Union's Seventh Framework Programme (FP7/2007-2013) ERC grant agreement no.\ 334828 and under the European Union's Horizon 2020 research and innovation programme (grant agreement No.~947778).
The paper reflects only the authors' views and not the views of the ERC or the European Commission. The
European Union is not liable for any use that may be made of the information contained therein.}

\address[Andreas Galanis, Leslie Ann Goldberg]{Department of Computer Science, University of Oxford, OX1 3QD, United Kingdom.}
\email{\{andreas.galanis,leslie.goldberg\}@cs.ox.ac.uk}

\address[Heng Guo]{School of Informatics, University of Edinburgh, Informatics Forum, Edinburgh, EH8 9AB, United Kingdom.}
\email{hguo@inf.ed.ac.uk}

\address[Kuan Yang]{John Hopcroft Center for Computer Science, Shanghai Jiao Tong University, 800 Dongchuan Road, Shanghai, 200240, China.}
\email{kuan.yang@sjtu.edu.cn}

\date{24th May 2021}

\begin{document}

\begin{abstract}
We give the first efficient algorithm to approximately count the number of solutions in the random $k$-SAT model when the density of the formula scales exponentially with $k$. The best previous counting algorithm for the permissive version of the model was due to Montanari and Shah and was based on the correlation decay method, which works up to densities $(1+o_k(1))\frac{2\log k}{k}$, the Gibbs uniqueness threshold for the model. Instead, our algorithm harnesses a recent technique by Moitra to work for random formulas
with much higher densities. The main challenge in our setting is to account for the presence of high-degree variables whose marginal distributions are hard to control and which cause significant correlations within the formula.  
\end{abstract}

\maketitle

\section{Introduction}

Let $\Phi=\Phi(k,n,m)$ be a $k$-CNF formula on $n$ Boolean variables with $m$ clauses chosen uniformly at random where each clause has size $k\ge 3$.
The random formula $\Phi$ shows an interesting threshold behaviour,
where the asymptotic probability that $\Phi$ is satisfiable drops dramatically from 1 to 0 when the density $\alpha:= m/n$ crosses a certain threshold $\alpha_{\star}$. There has been tremendous progress on establishing this phase transition and pinpointing the threshold $\alpha_{\star}$ \cite{KKKS98,FR99, AM02,AP03,CP16,DSS15} guided by elaborate but non-rigorous methods in physics \cite{MPZ02,MMZ05}. The exact value of the threshold $\alpha_{\star}$ is established in \cite{DSS15} for sufficiently large $k$; it is known that $\alpha_{\star}=2^k\ln 2-\tfrac{1}{2}(1+\ln 2)+o_k(1)$ as $k\rightarrow\infty$.

In contrast, the ``average case'' computational complexity of random $k$-CNF formulas remains elusive.
It is a notoriously hard problem to design algorithms that succeed in finding a satisfying assignment when the density of the formula $\Phi$ is close to (but smaller than) the satisfiability threshold $\alpha_{\star}$. The best polynomial-time algorithm to find a satisfying assignment of $\Phi$ is due to Coja-Oghlan \cite{Coj10}, which succeeds if $\alpha<(1-o_k(1))\cdot2^k\ln k/k$. It is known that beyond this density bound $ {2^k} \ln k/k$ the solution space of the formula undergoes a phase transition and becomes severely more complicated \cite{AC08}, so local algorithms are bound to fail to find a satisfying assignment in polynomial time (see for example \cite{Het16,Coj17,CHH17}). 

It is also a natural question to determine the number of satisfying assignments to $\Phi$,
denoted by $Z(\Phi)$, when the density is below the satisfying threshold.
It has been shown that $\tfrac{1}{n}\log Z(\Phi)$ is concentrated around its expectation \cite{AM14,CR13} for $\alpha < (1-o_k(1))\cdot2^k\ln k/k$.
However, for the random $k$-SAT model, there is no known formula for the expectation $\Ex \tfrac{1}{n}\log Z(\Phi)$
(though see \cite{2SAT} and \cite{SSZ16,CW18} for progress  along these lines for the case $k=2$ and for more symmetric models of random formulas, respectively).
Regarding the algorithmic question, Montanari and Shah~\cite{MS07} have given an efficient algorithm to approximate a closely related (permissive) version\footnote{\label{foot:one}More precisely, for the relevant densities $\alpha$,  \cite{MS07} gives, for all reals $\beta>0$, a PTAS for the partition function $\tfrac{1}{n}\log Z_\beta(\Phi)$ of a weighted model on all Boolean assignments $\sigma$, where the weight of an assignment $\sigma$ is $e^{-\beta H(\sigma)}$ and $H(\sigma)$ is the number of unsatisfied clauses under $\sigma$. The number of satisfying assignments $Z(\Phi)$ corresponds to $Z_\beta(\Phi)$ when   $\beta\rightarrow\infty$. Technically, the algorithm of \cite{MS07} does not yield an algorithm for $\beta=\infty$, even though it applies to arbitrarily large $\beta$ (in fact, even for $\beta\leq n^{\delta}$ for some small constant $\delta>0$).} of $\tfrac{1}{n}\log Z(\Phi)$ if $\alpha \le \frac{2\log k}{ k}(1+o_k(1))$,
based on the correlation decay method and the uniqueness threshold of the Gibbs distribution. Note that the threshold for the density~$\alpha$ is exponentially lower than the satisfiability threshold, and no efficient algorithm was known to give a more precise approximation to $Z(\Phi)$ (rather than the log-partition function).

In this paper, we address the algorithmic counting problem by giving the first \emph{fully polynomial-time approximation scheme} (FPTAS) for the number of satisfying assignments to random $k$-CNF formulas,
if the density $\alpha$ is less than $2^{rk}$, for sufficiently large $k$ and some constant $r>0$.
Our bound is exponential in $k$ and goes well beyond the uniqueness threshold of $\frac{2\log k}{k}(1+o_k(1))$
which is required by the correlation decay method.

Our  result is related to other algorithmic counting results on 
random graphs  
such as counting colourings, independent sets, and other structures in random graphs, see for example \cite{MS13,Eft,YZ16,EHSV18,LLLM19,Blanca}.
However, previous methods, such as Markov Chain Monte Carlo and Barvinok's method, appear to be difficult to apply to random formulas.
Instead, our algorithm is the first adaptation of Moitra's method \cite{Moi19} to the random instance setting.
We give a high level overview of the techniques in Section~\ref{sec:tech}.

\subsection{The model and the main result}

For $k\geq 3$,
let $\Phi=\Phi(k,n,m)$ denote a $k$-SAT formula chosen uniformly at random from the set of all $k$-SAT formulas with $n$ variables and $m$ clauses. 
Specifically, $\Phi$ has $n$ variables $v_1, v_2, \ldots, v_n$ and $m$ clauses $c_1, c_2, \ldots, c_m$. 
Each clause $c_i$ has $k$ literals $\ell_{i, 1}, \ell_{i, 2}, \ldots, \ell_{i, k}$ and each literal $\ell_{i, j}$ is chosen uniformly at random from $2n$ literals $\{v_1, v_2, \ldots, v_n, \lnot v_1, \lnot v_2, \ldots, \lnot v_n\}$. 
Note that each clause has exactly $k$ literals (repetitions allowed), 
so there are $(2n)^{km}$ possible formulas; we use $\Pr(\cdot)$ to denote the uniform distribution on the set of all such formulas. Throughout, we will assume that  $m=\left\lfloor n\alpha\right\rfloor$, where $\alpha>0$ is the density of the formula. 
We say that an event $\mathcal{E}$ holds \emph{w.h.p.}~if $\Pr(\mathcal{E})=1-o(1)$ as $n\rightarrow\infty$.

For a $k$-SAT formula $\Phi$, we let $\Omega=\Omega(\Phi)$ denote the set of satisfying assignments of $\Phi$.
If $\Omega$ is non-empty, we let
  $\pi_\Phi$ denote the uniform distribution on $\Omega$.
  
\newcommand{\statethmmain}{There is a polynomial-time algorithm $\mathcal{A}$ and there are two constants $r>0$ and $k_0 \geq 3$ 
  such that, for all $k\geq k_0$ and all $\alpha<2^{rk}$, the following holds w.h.p.~over the choice of the random $k$-SAT formula $\Phi=\Phi(k,n,\lfloor \alpha n \rfloor)$.
  The algorithm $\mathcal{A}$,   given as input the formula $\Phi$ and a rational $\eps>0$, outputs in time $poly(n,1/\eps)$ a number $Z$ that satisfies 
  $e^{-\epsilon}\aabs{\Omega(\Phi)} \leq Z \leq e^{\epsilon}\aabs{\Omega(\Phi)} $.}
\begin{theorem}\label{thm:main}
\statethmmain
 \end{theorem}

Throughout this paper, we will assume that $k\geq k_0$ where $k_0$ is a sufficiently large constant. 
We will also assume that the density $\alpha$ of the formula $\Phi$ satisfies $\alpha<2^{k/300}/k^3$, so $r$ can be taken to be $ {1}/{301}$ in 
Theorem~\ref{thm:main}. 
The constant $300$ here is not optimised, but we do not expect to be able to use the current techniques to improve it substantially.  Our main point   is that for a density which is exponential in $k$, an FPTAS exists for random $k$-CNF formulas.
Finally, we assume that $k^2\alpha\geq 1$, otherwise it is well-known (see, e.g., Theorem~$3.6$ in \cite{SS85}) that w.h.p.~every connected component of $\Phi$, 
viewed as a hypergraph where variables correspond to vertices and clauses correspond to hyperedges, is of size $O(\log n)$. 
In this case we can count the number of satisfying assignments by brute force.

\subsection{New developments} After this paper was written, Coja-Oghlan, M\"{u}ller, and Ravelomanana \cite{BPRandom} have given a conditional formula  for the permissive version (analogously to \cite{MS07}, cf.~Footnote~\ref{foot:one}) of the log-partition function  $\mathbb{E}\tfrac{1}{n}\log Z(\Phi)$. Their condition captures a non-reconstruction property of the Gibbs distribution that is believed to hold up to densities $(1-o_k(1))2^k\ln k/k$ (though this is not known). As mentioned earlier, our algorithmic result applies to the non-permissive version and gives unconditional approximation guarantees for the partition function $Z(\Phi)$ with arbitrarily small relative error. For future work, it would be interesting to see whether the probabilistic approach of \cite{BPRandom} can be combined with the algorithmic perspective of this paper.

\subsection{Algorithm overview}
\label{sec:tech}

We give a high-level overview of our algorithm here before  giving the details.
Approximately counting the satisfying assignments of a $k$-CNF formula 
has been a challenging problem  using traditional algorithmic techniques,
since the  solution space (the set of satisfying assignments)  is complicated and it is 
not connected, using the transitions of  commonly-studied Markov chains. 
Recently some new approaches were introduced \cite{Moi19,GJL19}. 
Most notably, the breakthrough work of Moitra \cite{Moi19} gives the first (and so far the only) efficient deterministic algorithm 
that can approximately count the  satisfying assignments of $k$-CNF formulas  in which each variable appears in at most $d$ clauses, if, roughly speaking, $d\lesssim 2^{k/60}$.
Inspired by this, Feng et al.~\cite{NewHeng} have also given a MCMC algorithm which applies when $d\lesssim 2^{k/20}$.
Subsequent to this paper, Jain et al.~\cite{Jain} improved this bound even
further, showing how to approximately count when $d \lesssim 2^{k/5.741}$.

As our goal is to count satisfying assignments of sparse random $k$-CNF formulas,
where these degree bounds do not hold, but average degrees are small,
it is natural to also choose Moitra's method in the random instance setting.
However, the first difficulty is that Moitra's method  
relies on the fact that  the marginal probability of each variable
(the probability that the variable is true in a uniformly-chosen satisfying assignment) is
nearly~$1/2$.
This is necessary because Moitra's method involves solving a certain linear program (LP) and the size of this LP
is polynomially-bounded only if a certain process couples quickly.
The proof that the process couples quickly relies on the fact that the marginals are nearly~$1/2$ (and certainly on the fact that they are bounded away from~$0$ and~$1$).
In contrast,   for a random $k$-CNF formula,
although the \emph{average} degree of   variables is low,
with high probability 
there are  variables with degrees as high as $\Omega\left( {\log n}/{\log\log n} \right)$.
In the presence of these high-degree variables,
it is no longer true that the marginal probabilities of the variables are nearly~$1/2$.
In fact, they can be arbitrarily  near~$0$ or~$1$.

Our solution to this issue is to separate out high-degree variables, as well as those that are heavily influenced by high-degree variables.
To do this, we define a process to recursively label ``bad'' variables.
At the start, all high-degree variables are bad.
Then, all clauses containing more than $k/10$ bad variables are labelled bad, as  are all variables that they contain.
We run this process until no more bad clauses  are found. We call the remaining variables and clauses of the formula ``good''. A key property is that 
all good variables have  an upper bound on their degree and all good clauses  contain at least $9k/10$ good variables; this allows us to show that the marginal probabilities of good variables are close to 1/2.

The next step is to attempt to apply Moitra's method. 
The goal of Moitra's method is to compute more precise estimates for the marginal probabilities of variables; given  
accurate estimates on the marginal probabilities it is then relatively easy to approximate the number of satisfying assignments using refined self-reducibility techniques.

Of course, we need to modify the method to deal with the bad variables, which still appear in the formula.
We  first explain Moitra's method   and then proceed with our modifications.
The first step is to mark variables, so that every clause contains a good fraction of marked variables
and a good fraction of  unmarked variables.
Then, for a particular marked variable $v$, 
we set up  an LP. As noted earlier, 
the variables of the LP correspond to the states of a certain coupling process
which couples two distributions on satisfying assignments using the marked variables --- the first distribution over satisfying assignments in which $v$ is true,
and the second distribution over satisfying assignments in which $v$ is false.
Solving the LP  recovers the transition probabilities of the coupling process and yields enough information to  approximate the marginal probability of $v$.

In order to guarantee that the size of the LP is bounded by a polynomial in the size of the original CNF formula,
we have to restrict the coupling process. The process can be viewed as a tree 
 and it suffices to truncate this tree at a suitable level.

Thus, a crucial part of the proof (both in Moitra's case and in ours) is to show that the error caused by the truncation is sufficiently small.
The reason that the error caused by the truncation is small is that, with high probability, branches of the coupling tree
``die out'' before reaching a large level. The reason for this is that the marginals of marked variables stay near~$1/2$, even when conditioning on partial assignments.

 In our case where $\Phi$ is a random formula, the marginals are not all near~$1/2$, even without
 any conditioning.  
But the good variables do have marginals near~$1/2$.
 So we only mark/unmark good variables  and we ``give up'' on bad variables.
 Given that we don't have any control over the bad variables, we have to modify the coupling process. 
Thus, whenever we meet a bad variable in the coupling process, we have to assume the worst case and treat this variable and all bad variables connected to it 
as if they all have failed the coupling, meaning that the disagreement spreads quickly over bad components.

The most important part of our analysis is to upper bound the size of connected bad components and how often we encounter them during the coupling process.
Given these upper bounds, we are able to show that the coupling still dies out sufficiently quickly, so the
error  caused by the truncation is not too large.

Solving the LP then allows us to estimate the marginals of the good variables. Given that the bad components have small size, this
turns out to be enough information to allow us to estimate the number of satisfying assignments of the original formula (containing both 
good and bad clauses).

We conclude this summary by discussing the prospects for improving our work.
Although we have given an efficient algorithm which works for densities that are exponentially large in $k$,
the densities that we can handle are still small compared to the satisfiability threshold or to the threshold under which efficient search algorithms exist.
Perhaps a modest start towards obtaining comparable thresholds 
for approximate counting algorithms would be to 
  consider  models whose state spaces are connected.
For example, for monotone $k$-CNF formulas where each variable appears in at most $d$ clauses, 
Hermon et al.\ \cite{HSZ19} showed that efficient randomised algorithms exist if $d\le c2^{k/2}$ for some constant $c>0$, 
which is optimal up to the constant $c$ due to complementing hardness results \cite{BGGGS19}. 
They also showed that the same algorithm works for \emph{random regular} monotone $k$-CNF formulas, if the degree $d\le c2^k/k$ for some $c>0$.
It remains open whether an average case bound of the same order can be achieved for random monotone $k$-CNF formulas.

\section{Notation}

To help keep track of the notation defined in this paper, the reader is referred to the table in Appendix (Section~\ref{sec:app}).

\section{High-degree and bad variables}
\label{sec:high-degree}

We will apply the method of Moitra \cite{Moi19}, which 
was introduced to approximately count the  satisfying
assignments of $k$-CNF formulas in which each literal appears a bounded
number of times. The main difference between the formulas studied
by Moitra and the random formulas that we study  is that, in our formulas,  some variables 
will occur many more times than the average. 

\begin{definition}\label{def:HD}
  Let $\Phi$ be a $k$-SAT formula. We say that a variable $v$ of~$\Phi$ is \emph{high-degree} if $\Phi$ contains at least $\HighD:=2^{k/300}$ 
  occurrences of literals involving the variable~$v$.
\end{definition}

 In our algorithm, we will not be able to control  these high degree variables 
 or other variables that are  affected by them.
 These variables contribute to the ``bad'' part of the formula~$\Phi$. Formally, denote the set of clauses of $\Phi$ by $\Cl$ and the set of variables by $\Va$.
For each $c\in \Cl$, let
 $\var(c)$ denote the set of variables in~$c$.
 For each subset~$C$ of~$\Cl$, let  $\var(C):=\cup_{c\in C}\var(c)$.
The \emph{bad variables} and \emph{bad clauses} of $\Phi$ are identified  by running the following process:
\goodbreak
\begin{enumerate}
  \item $\Va_0$ (the initial bad variables) $\gets$ the set of high-degree variables;
  \item $\Cl_0$ $\gets$ the set of clauses with at least $k/10$~variables in~$\Va_0$;
  \item $i\gets 0$;
  \item Do the following until $\Va_i=\Va_{i-1}$:
  \begin{itemize}
    \item $i\gets i+1$;
    \item $\Va_i\gets \Va_{i-1}\cup\var(\Cl_{i-1})$;
    \item $\Cl_i \gets$ $\{c\in\Cl\mid \var(c)\cap \Va_i\ge k/10\}$;
  \end{itemize}
  \item $\Cl_\bad \gets \Cl_i$ and $\Va_\bad \gets \Va_i$;
  \item $\Cl_\good \gets \Cl \setminus \Cl_i$ and $\Va_\good \gets \Va\setminus \Va_i$.
\end{enumerate}
The method that we have used to define good variables and clauses is inspired by~\cite{CP16}.

\begin{observation}\label{obs:goodbad} 
For any $c\in\Cl_\good$, $\abs{\var(c)\cap\Va_\bad}<k/10$. For any $c\in\Cl_\bad$, $\abs{\var(c)\cap\Va_\good}=0$.
\end{observation}

The formula $\Phi$ naturally corresponds to a bipartite ``factor graph'' where one side is variables and the other clauses.
We will use the following two graphs $G_\Phi$ and $H_\Phi$ which are the dependency graph induced by the factor graph on clauses and variables, respectively.

\begin{definition}\label{def:Gphi}
  Let $G_\Phi$ be the graph with vertex set $\Cl$ in which two clauses $c$ and $c'$ are adjacent if and only if $\var(c) \cap \var(c')\neq\emptyset$.
  We say that a set $C\subseteq \calC$ of clauses is connected if the induced subgraph $G_\Phi[C]$ is connected.
Let $G_{\Phi,\good}$ be the graph with vertex set $\Cl_\good$ in which two clauses $c$ and $c'$ are adjacent
  if and only if $\var(c) \cap \var(c')\cap \Va_\good\neq\emptyset$, i.e., $c$ and $c'$ share a good variable.
\end{definition}

Since each $v\in\Va_\good$ has at most $\HighD$ literal occurrences,
the maximum degree in $G_{\Phi,\good}$ is at most $k(\HighD-1)$.

\begin{definition}\label{def:Hphi}
  Let $H_\Phi$ be the graph with vertex set $\Va$ in which two variables $v$ and $v'$ are adjacent 
  if and only if there exists a clause $c\in \Cl$ such that $v,v'\in \var(c)$. 
  We say that a set $V\subseteq\Va$ of variables is connected if the induced subgraph $H_\Phi[V]$ is connected.
Let $H_{\Phi,\bad}$ be the graph with vertex set $\Va_\bad$ in which two variables $v$ and $v'$ are adjacent
  if and only if there exists a bad clause $c\in \Cl_{\bad}$ such that $v,v'\in \var(c)$. 
  We say that a set $V\subseteq\Va$ of variables is a \emph{bad component} if $V$ is a connected component in $H_{\Phi,\bad}$.
\end{definition}

\section{The \Lovasz\ local lemma}
\label{sec:local-lemma}

The \Lovasz\ local lemma \cite{EL75} is an important tool for our algorithm.
In particular, we will need the asymmetric version (proved by \Lovasz\ and published in \cite{Spe77}).
For convenience, we will specialize the local lemma to the variable setting.

\begin{definition}\label{def:lll}
Let $\Omega^*$ be the set of all $2^{|\Va|}$ assignments $\Va\to \{\true,\false\}$.
Given any subset $A \subseteq \Omega^*$,
let $\mu_A$ be the uniform distribution on $A$. 
We say that a subset $A\subseteq \Omega^*$ depends on a variable $v\in \Va$
if there exist $\sigma\in A$ and $\sigma' \in \Omega^*\setminus A$ such that $\sigma$ and $\sigma'$ differ only at~$v$.
We use $\var(A)$ to denote the set of variables on which $A$ depends.
Let $\Gamma(A) = \{c \in \Cl \mid  \var(c) \cap \var(A) \neq \emptyset\}$  
and $\Gamma_\good(A) = \{c\in\calC\mid \var(A)\cap \var(c)\cap \Va_\good \neq\emptyset \}$. Similarly, for any $c\in \Cl$, let $\Gamma(c) = \{c' \in \Cl \mid  \var(c) \cap \var(c') \neq \emptyset\}$  
and $\Gamma_\good(c) = \{c'\in\calC\mid \var(c)\cap \var(c')\cap \Va_\good \neq\emptyset \}$.
\end{definition}

When reading the following theorem, it is helpful to think of $A_c$ as being the set of assignments that fail to
satisfy the clause~$c$ (though we will also use the theorem in other ways).

\begin{theorem}[The local lemma]\label{thm:lll}
    For each $c\in \Cl$, let $A_c$ be a subset of $\Omega^*$ such that   $\var(A_c)\subseteq\var(c)$.
    If there exists a function $x\colon \Cl \to (0,1)$ such that, for all $c\in \Cl$,
    \begin{align}
    \Pr_{\mu_{\Omega^*}}(A_c) \leq x(c)\prod_{b\in \Gamma(A_c)}\Bigl(1-x(b)\Bigr),
    \label{eqn:lll-condition}
    \end{align}
    then
    \begin{align*}
    \Pr_{\mu_{\Omega^*}}\Bigl( \bigwedge_{c\in\Cl} \complement{A_c} \Bigr)>0.
    \end{align*}
    Furthermore, for any $A\subseteq \Omega^*$,
    $\Pr_{\mu_{\Omega^*}}\Bigl( A\mid \bigwedge_{c\in\Cl} \complement{A_c} \Bigr)
    \leq \Pr_{\mu_{\Omega^*}}(A)\prod_{b\in \Gamma(A)}\Bigl(1-x(b)\Bigr)^{-1}$.
\end{theorem}

The second part of Theorem~\ref{thm:lll}
is due to
Haeupler et al.~\cite{HSS11}. It provides an upper bound on the probability of
an event under the uniform distribution over satisfying assignments.
It is also possible to   find an assignment such that $\bigwedge_{c\in\Cl} \complement{A_c}$ holds in $O(\abs{\Va})$ time
by the algorithm of Moser and Tardos \cite{MT10}.

In the course of our algorithm it is useful to ``mark'' some of the variables in $\Va_\good$.
For this, we use the approach of Moitra~\cite{Moi19}.
Formally, a ``marking'' is an assignment   from $\Va_\good$ to $\{\text{marked},\text{unmarked}\}$.
All bad variables are unmarked.
Using Theorem~\ref{thm:lll}, we have the following lemma.

\begin{lemma}  \label{lem:marking}
W.h.p.{} over the choice of~$\Phi$, there exists a marking on $\Va_\good$ such that:
  \begin{enumerate}
  \item \label{item:marked} every good clause has at least $3k/10$ marked variables and at least $k/4$ unmarked good variables;
  \item \label{item:partial} there is a partial assignment of bad (and thus unmarked) variables that satisfies all bad clauses.
  \end{enumerate}
  Furthermore, such a marking can be found in deterministic polynomial time.
\end{lemma}
\begin{proof}
  We apply Theorem~\ref{thm:lll} on $\Cl_\good$ and $\Va_\good$.
  Let $\Omega^*_\good$ be all possible markings: $\Va_\good\to \{\text{marked},\text{unmarked}\}$,
  so $\mu_{\Omega^*_\good}$ is the distribution
  in which each good variable is marked independently and uniformly at random.
  For $c\in\Cl_\good$, let
  $M$ be the number of marked good variables in $\var(c)$.
  Let $A_c$ be the subset of $\Omega^*_\good$ that $M<3k/10$ or $M> 3k/5$. 
  Since $c\in \Cl_\good$, $c$ contains at least $9k/10$ good variables
  (see the observation immediately following the process that defines bad variables and clauses).
  Thus $9k/20\le \Ex_{\mu_{\Omega^*_\good}} M\le k/2$, and the number of unmarked good variables in $\var(c)$ is at least $9k/10 - M$.
  By a Chernoff bound,
  \begin{align*}
    \Pr_{\mu_{\Omega^*_\good}}(A_c) & = \Pr_{\mu_{\Omega^*_\good}}(M>3k/5)+\Pr_{\mu_{\Omega^*_\good}}(M<3k/10) \\
    & \le e^{-k/110}+e^{-k/40}\le 2e^{-k/150}.
  \end{align*}
  Let $x(c)=\frac{1}{k\HighD}$ if $c\in\Cl_\good$.
  Since we only consider assignments of $\Va_\good$, so $\var(A_c) \subseteq \var(c) \cap \Va_\good$ and thus $\Gamma(A_c) \subseteq \Gamma_\good(c)$.
  Note that $\Gamma_\good(c)$ is the set of neighbours of $c$ in $G_{\Phi,\good}$, and the maximum degree in $G_{\Phi,\good}$ is at most $k(\HighD-1)$,
  so we can verify \eqref{eqn:lll-condition} as follows,
  for any $c\in\Cl_\good$ and sufficiently large $k$,
  \begin{align*}
    x(c)\prod_{b\in \Gamma(A_c)}\Bigl(1-x(b)\Bigr) & \ge x(c)\prod_{b\in \Gamma_\good(c)}\Bigl(1-x(b)\Bigr) \ge \frac{1}{k\HighD}\left(1-\frac{1}{k\HighD}\right)^{k\HighD} \ge \frac{1}{e^2 k\HighD}\\
  & \geq e^{-k/300-2}/k > 2e^{-k/150} \ge \Pr_{\mu_{\Omega^*_\good}}(A_c).
  \end{align*}
  Thus, there is a marking such that item \eqref{item:marked} holds.
  
W.h.p.{} over the choice of~$\Phi$, item \eqref{item:partial} always holds for any marking of $\Va_\good$.
  This is because our density is well below the critical threshold \cite{DSS15}.
  Thus there is at least one satisfying assignment to $\Phi$.
  As $\var(c)\cap\Va_\good=\emptyset$ for any $c\in\Cl_\bad$, 
  the restriction of the satisfying assignment to $\Va_\bad$ satisfies the condition.

  The marking can be found using the deterministic algorithm \cite{CGH13} by verifying  
  \begin{align*}
    \Pr_{\mu_{\Omega^*_\good}}(A_c) & \le \left(x(c)\prod_{b\in \Gamma_\good(c)}\Bigl(1-x(b)\Bigr)\right)^{1.01}  \qedhere
  \end{align*} 
  \end{proof}

We will assume from now on that the random formula~$\Phi$ satisfies Lemma~\ref{lem:marking}
and we 
will stick to an arbitrary marking given by Lemma~\ref{lem:marking}.
We use $\marked(c)$ to denote the marked variables in clause~$c$ 
and $\Va_\marked:=\cup_{c\in\Cl_\good}\marked(c)$ be the set of all marked variables.

Let $\Omega$ be the set of satisfying assignments of~$\Phi$.
We will be particularly interested in the uniform distribution $\mu_{\Omega}(\cdot)$.
For any partial assignment $\partialVal$ of some of the \emph{good} variables of~$\Phi$,
let $\Phi^{\partialVal}$ be the formula produced by simplifying $\Phi$   under $\partialVal$.
In other words, we remove all  clauses that are satisfied under $\partialVal$ and 
we remove all false literals  from all clauses.
(Some clauses may become empty in $\Phi^{\partialVal}$, 
in which case $\Phi^{\partialVal}$ cannot be satisfied.)
We use $\Cl^{\partialVal}$ to denote the set of clauses of $\Phi^{\partialVal}$,
and similarly $\Va^{\partialVal}$. We also define $\Va_\good^{\partialVal} = \Va_\good \cap \Va^{\partialVal}$ and $\Cl_\good^\partialVal = \Cl_\good \cap \Cl^{\partialVal}$ to denote the sets of remaining good variables and clauses of $\Phi$ simplified under $\partialVal$. (Remark: Note that $\partialVal$ does not contain bad variables so if we define $\Va_\bad^\partialVal$ and $\Cl_\bad^\partialVal$ similarly then we will have $\Va_\bad^\partialVal = \Va_\bad$ and $\Cl_\bad^\partialVal=\Cl_\bad$.)
Let $\Omega^{\partialVal}$ denote the set of satisfying assignments of $\Phi^{\partialVal}$,
namely, those satisfying assignments of $\Phi$ consistent with~$\partialVal$.

Let $s:= 2^{k/4}/(ek\HighD)$. Under an arbitrary conditioning of marked variables,
we have good control of events defined by good variables,
and the marginal distribution of good variables.

\begin{lemma}  \label{lem:local-uniform}
  Let $V\subseteq\Va_\marked$.
  For any partial assignment $\partialVal:V\rightarrow\{\true,\false\}$ such that $\Omega^\partialVal\neq\emptyset$ and
  any subset $A\subseteq \Omega^\partialVal$ such that $\var(A)\subseteq\Va_\good\setminus V$,
  \begin{align*}
    \Pr_{\mu_{\Omega^{\partialVal}}}(A)\le \Pr_{\mu_\calA}(A)\Bigl(1-\frac{1}{sk\HighD}\Bigr)^{-\abs{\Gamma(A)}},
  \end{align*}
  where $\calA$ is the set of all assignments of $\var(A)$.
  In particular, for $v\in\Va_\marked\setminus V$,
  \begin{align*}
    \left(1-\frac{1}{3s}\right)\frac{1}{2}\le\Pr_{\mu_{\Omega^{\partialVal}}}(v\rightarrow \true)\le \left(1+\frac{1}{3s}\right)\frac{1}{2}.
  \end{align*}
\end{lemma}
\begin{proof}
  Let $\partialVal_\bad$ be an arbitrary assignment on $\Va_\bad$ such that all bad clauses are satisfied.
  Such a $\partialVal_\bad$ exists because of Item \eqref{item:partial} of Lemma~\ref{lem:marking}.

  We apply Theorem~\ref{thm:lll} to $\Phi^{\partialVal\cup\partialVal_\bad}$.
  The remaining variables are $V'=\Va_\good\setminus V$ and clauses are $\Cl':=\Cl_\good^{\partialVal\cup\partialVal_\bad}$.
  Let $\Omega'$ be the set of all $2^{|V'|}$ assignments $V'\to \{\true,\false\}$.
  Let $A_c$ be the subset of $\Omega'$ that $c$ is not satisfied, for $c\in\Cl'$.
  Then,
  \begin{align}\label{eqn:distribution-equiv}
    \Pr_{\mu_{\Omega^{\partialVal}}}\bigl(\cdot\mid\partialVal_\bad\bigr) = \Pr_{\mu_{\Omega'}}\bigl(\cdot\mid\bigwedge_{c\in\Cl'}\complement{A_c}\bigr).
  \end{align}

  We verify the condition \eqref{eqn:lll-condition} of Theorem~\ref{thm:lll} and show
  \begin{align}\label{eqn:lll-good-event}
    \Pr_{\mu_{\Omega'}}\bigl(A\mid\bigwedge_{c\in\Cl'}\complement{A_c}\bigr)\le \Pr_{\mu_\calA}(A)\Bigl(1-\frac{1}{sk\HighD}\Bigr)^{-\abs{\Gamma(A)}},
  \end{align}
  and
  \begin{align}\label{eqn:condition-local-uniform}
    \left(1-\frac{1}{3s}\right)\frac{1}{2}\le\Pr_{\mu_{\Omega'}}\bigl(v\rightarrow \true\mid\bigwedge_{c\in\Cl'}\complement{A_c}\bigr)\le \left(1+\frac{1}{3s}\right)\frac{1}{2}.
  \end{align}
  The first part of the lemma follows directly
  and the second part follows from \eqref{eqn:distribution-equiv},  
  \eqref{eqn:condition-local-uniform}, and the law of total probabilities over all choices of $\partialVal_\bad$.

  By Lemma~\ref{lem:marking}, there are at least $k/4$ unmarked variables in each $c\in\Cl'$.
  Thus,
  \begin{align*}
    \Pr_{\mu_{\Omega'}}(A_c)\le 2^{-k/4}.
  \end{align*}
  Let $x(c)=\frac{1}{sk\HighD}$. Since $\Omega'$ is a set of assignments of good variables, we have $\var(A_c) \subseteq \var(c)\cap \Va_\good$, and thus $\Gamma(A_c) \subseteq \Gamma_\good(c)$.
  Again, note that $\Gamma_\good(c)$ is the set of neighbours of $c$ in $G_{\Phi,\good}$ and the maximum degree of $G_{\Phi,\good}$ is at most $k(\HighD-1)$, so we obtain
  \begin{align*}
  x(c)\prod_{b\in \Gamma(A_c)}\Bigl(1-x(b)\Bigr)  &\geq x(c)\prod_{b\in \Gamma_\good(c)}\Bigl(1-x(b)\Bigr) \ge \frac{1}{sk\HighD}\left( 1-\frac{1}{sk\HighD} \right)^{k(\HighD-1)} \cr
 & \ge\frac{e^{-1/s}}{sk\HighD}\ge \frac{1}{esk\HighD}=2^{-k/4}.
  \end{align*}
  The two inequalities above verify condition \eqref{eqn:lll-condition} of Theorem~\ref{thm:lll}.
  Thus, \eqref{eqn:lll-good-event} follows directly since $\var(A)\subseteq V'$.
  Moreover, setting $A=v\rightarrow\true$, for sufficiently large $k$,
  \begin{align*}
    \Pr_{\mu_{\Omega'}}\bigl(v\rightarrow \true\mid\bigwedge_{c\in\Cl'}\complement{A_c}\bigr) 
    &\le \frac{1}{2}\left( 1-\frac{1}{sk\HighD} \right)^{-\HighD} \le \frac{1}{2} \cdot e^{\HighD/(sk\HighD-1)} \le \frac{1}{2}\left( 1+\frac{1}{3s} \right).
  \end{align*}
  We get the same upper bound for the event $v\rightarrow\false$ by the same argument.
  The bound \eqref{eqn:condition-local-uniform} follows by combining these two bounds.
\end{proof}

Moreover, we have the following lemma for a partial assignment that we will use to  apply self-reducibility.

\begin{lemma}\label{lem:partial-assignment}
  Let $\Phi=\Phi(k, n, m)$ and let $v_1,v_2,\ldots, v_n$ be the   variables  of $\Phi$.  In each clause,  order the literals in the order 
  induced
  by the indices of their variables.
  Then there is a partial assignment $\partialVal^*$  of truth values to some subset of $\Va_\marked$ 
   with the property that every clause $c\in\Cl_\good$ is satisfied by its first $k/20$   literals corresponding to marked variables.
  Moreover, $\partialVal^*$ can be found in deterministic polynomial time.
\end{lemma}
\begin{proof}
  We replace every $c\in\Cl_\good$ by its first $k/20$ literals corresponding to marked variables.
  Call the new clause $c'$ and the new set $\Cl_\good'$.
  This induces a new formula $\Phi'$ whose clause set is $\Cl_\good'$ with no unmarked variables. 
  We apply Theorem~\ref{thm:lll} to $\Phi'$.
  Let $A_c$ be the event that $c$ is not satisfied where $c\in\Cl_\good'$ and set $x(c)=\frac{1}{k\HighD}$ for all $c\in\Cl_\good'$.
  Since $\Phi'$ is a smaller formula, the size of $\Gamma_\good(c)$ is still at most $k(\HighD-1)$.
  It is straightforward to verify that the condition \eqref{eqn:lll-condition} holds and Theorem~\ref{thm:lll} applies.
  Thus, the desired $\partialVal^*$ exists.

  To find $\partialVal^*$, once again we apply the deterministic algorithm for the local lemma \cite{CGH13}.
\end{proof}

Recall $\Phi^{\partialVal}$, which is $\Phi$ simplified under $\partialVal$.
Under $\partialVal^*$, neither bad clauses nor bad variables will be removed.
More importantly, we have the following corollary.

\begin{corollary}\label{cor:partial-prefix}
  For any prefix $\partialVal$ of $\partialVal^*$ in Theorem~\ref{lem:partial-assignment},
  any remaining good clause $c$ in $\Phi^{\partialVal}$ satisfies $\marked(c)\ge k/4$.
\end{corollary}
\begin{proof}
  Under $\partialVal$, for any $c\in\Cl_\good$,
  either $c$ has at least $k/20$ marked variables assigned, or $c$ has at most $k/20$ marked variables assigned.
  In the first case, by Lemma~\ref{lem:partial-assignment}, $c$ is satisfied and thus is removed in $\Phi^{\partialVal}$.
  In the second case, even if $c$ is not satisfied, by Lemma~\ref{lem:marking}, $c$ has at least $(3/10-1/20)k = k/4$ marked variables unassigned.
\end{proof}

\section{The coupling tree}\label{sec:couplingtree}

W.h.p.{} the random formul~$\Phi$ satisfies Lemma~\ref{lem:marking} and
Lemma~\ref{lem:partial-assignment} and from now on we 
focus on~$\Phi$ for which this is true.
We will use the marking from Lemma~\ref{lem:marking}.
Conditioned on a prefix $\partialVal$ of the partial assignment $\partialVal^*$ from Lemma~\ref{lem:partial-assignment},
the main subroutine of our algorithm is to calculate the marginal probability of the next variable in $\partialVal^*$, say $v^*$.
Let $\Omega_1^\partialVal = \Omega^{\partialVal\cup\{v^*\rightarrow\true\}}$ 
and $\Omega_2^\partialVal = \Omega^{\partialVal\cup\{v^*\rightarrow\false\}}$. 
Then our goal is to estimate $|\Omega_1^\partialVal|/|\Omega_2^\partialVal|$.

We will eventually set up a linear program to approximate $|\Omega_1^\partialVal|/|\Omega_2^\partialVal|$.
Before introducing the linear program, we will define the so-called coupling tree,
which represents a variable-by-variable greedy coupling process between $\mu_{\Omega_1^\partialVal}$ and $\mu_{\Omega_2^\partialVal}$. Denote by $\=T^\partialVal$ the coupling tree.
Each node~$\rho$ of $\=T^\partialVal$ is a tuple $\rhois$ where
\begin{itemize} 
  \item $\VI$ and $\Vset$ are subsets of $\Va$;
  \item $\Crem$ and $\failedclauses$ are subsets of $\Cl$;
  \item $\Aone$ and $\Atwo$ are functions from~$\Vset$ to~$\{\true,\false\}$;
  \item $\reason$ is a function from $\Cl$  to the subsets of $\{\badsf,\disagree,\mathsf{1},\mathsf{2}\}$.
\end{itemize}
We use the notation $\Aone(\rho)$, $\Atwo(\rho)$, etc.\ to denote components of $\rho$.
Note that the information in $\rho$ is redundant, as some components can be deduced from others. 

Intuitively, for a node $\rho$, $\Vset$ is the set of variables that we have tried to couple, and $\VI$ is the
set of variables that cause discrepancies, or for which we have given up. 
The set $\Crem$ denotes remaining clauses, and $\failedclauses$ denotes ``failed'' clauses.
The two partial assignments $\Aone$ and $\Atwo$ assign truth values to the variables in~$\Vset$.
Finally, the function $\reason(c)$ gives the reason why $c$ is in $\failedclauses$.
The possibilities are: (1) $c$ is a bad clause;
(2) $\Aone$ and $\Atwo$ disagree on some variable in this clause;
(3) $c$ is not satisfied by the partial assignment $\Aone$  or $\Atwo$ (or both).
These reasons may not be mutually exclusive.
If $c\not\in \failedclauses$, then $\reason(c)=\emptyset$. 

We will only consider a partial assignment $\partialVal$ that is a prefix of $\partialVal^*$ from Lemma~\ref{lem:partial-assignment}.
Thus, by Corollary~\ref{cor:partial-prefix}, in any $c\in\Cl_\good^\partialVal$, there are at least $k/4$ marked variables remaining.

We will (inductively) guarantee that every node $\rhois$ of the coupling tree satisfies the following properties:
\begin{enumerate}[(P1)] 
\item $\{v^*\} \subseteq \Vset\cap \VI$. \label{P:xinVset}

\item Every clause $c\in \Crem$ satisfies 
 one of the following: $\var(c) \subseteq \VI$, $\var(c) \subseteq \Va^{\partialVal}\setminus \VI$, or $\marked(c) \setminus \Vset$
 is non-empty. \label{P:uselessclausesgone}
   \item For any $c\in\Crem\cap\Cl_\bad$, 
    $\var(c) \subseteq \Va^{\partialVal}\setminus \VI$. \label{P:badclause}
  \item For any $c\in \Cl^{\partialVal}\setminus \Crem$, at least one of the following is true:
  \begin{itemize}
    \item $c$ is satisfied by both of the partial assignments $\Aone$ and $\Atwo$;
    \item  $\var(c)\subseteq \VI\cup\Vset$. 
   \end{itemize} \label{P:rhoinv}
  \item $H_\Phi[\VI]$ is connected.\label{P:connected}
  \item For any $u\in \VI$,  
    $\exists c\in\failedclauses$ 
  s.t.\ $u\in\var(c)$. \label{P:VIvar}
  \item For any $u \in \Vset\setminus \VI$, $\Aone(u)=\Atwo(u)$.\label{P:agree}
  \item $\reason(c)$ is non-empty if and only if $c\in \failedclauses$.\label{P:failempty}
  \item \quad
  \begin{enumerate}[(P9.1)]
\item For any  $c \in \failedclauses$,  $\badsf \in \reason(c)$ only if $c$ is a bad clause; \label{P:typebad}
\item For any  good $c \in \failedclauses$,  \label{P:typedisagree}
$\disagree \in \reason(c)$ iff there is a variable  
$v \in \var(c) \cap \VI\cap \Vset$  
such that $\Aone(v) \neq \Atwo(v)$; 
\item For any  $c \in \failedclauses$ and  $i\in\{\mathsf{1},\mathsf{2}\}$, $i\in\reason(c)$ 
only if   the following three conditions hold: 
 $\var(c) \subseteq \VI\cup \Vset$,   
$\marked(c) \subseteq \Vset$, and $\Ai$ does not satisfy~$c$. 
\label{P:typetwoclauses}
 \end{enumerate}

\item  For any    
$c\in\Cl_\good^\partialVal$, $\marked(c)\cap \VI \subseteq \Vset$.
\label{P:yucktwo}

\end{enumerate}

The root of the coupling tree is the node~$\rho^*$
with 
$\Vset(\rho^*) = \VI(\rho^*) = \{v^*\}$, 
$\Crem(\rho^*) = \Cl^{\partialVal}$, and
$\failedclauses(\rho^*)$ assigned to the set of clauses containing~$v^*$.
The assignment $\Aone(\rho^*)$ sets $v^*$ to~$\true$ and
the assignment $\Atwo(\rho^*)$ sets $v^*$ to~$\false$.  
The function $\reason$ maps every clause in $\failedclauses(\rho^*)$
to $\{\disagree\}$ and every other clause to $\emptyset$.

It is straightforward to see that $\rho^*$ satisfies the coupling-tree properties.
Property~\ref{P:badclause} follows from the fact that $v^*$ is good, so it is not in a bad clause (Observation~\ref{obs:goodbad}).
To see Property~\ref{P:uselessclausesgone} note that each 
bad clause~$c$ has 
$\var(c) \subseteq \Va^{\partialVal}\setminus \VI$
and each good
clause has at least $k/4>1$ marked variables
so  $\marked(c) \setminus \Vset(\rho)$ 
is non-empty.  
The rest are straightforward.

In order to ensure that the size of the coupling tree is bounded from above by a polynomial in the size of the formula~$\Phi$,
we will set a truncation depth $L:= C_0 {(3k^2\HighD)} \depthceiling$ for some sufficiently large absolute constant $C_0$,
where $n=\abs{\Va}$.

\begin{definition}[leaf, truncating node]\label{def:leaf}
A node $\rho$ of the coupling tree is \emph{a leaf} if
$\abs{\VI(\rho)}\leq L$ and  every $c\in \Crem(\rho)$ has the property that
{$\var(c) \subseteq \VI(\rho)$}
or 
{$\var(c) \subseteq \Va^{\partialVal}\setminus \VI(\rho)$.}
If $\abs{\VI(\rho)} > L$, then $\rho$ is a \emph{truncating node}.
We denote the set of leaves by~$\calL$, the set of truncating nodes by~$\calT$,
and their union by $\calL^*:= \calL \cup \calT$.
\end{definition}

Suppose that $\rho$ is not a leaf or a truncating node. Then we define its children as follows. 
Since $\rho$ is not a leaf, 
there is a clause 
{$c\in \Crem(\rho)$} 
such that $\var(c)\cap\VI(\rho)\neq\emptyset$ and $\var(c)\cap (\Va^{\partialVal}\setminus \VI(\rho))\neq\emptyset$.
By \ref{P:badclause}, $c$ must be a good clause.
Let $c$ be the first such clause.
By \ref{P:uselessclausesgone},
let $u$ be the first variable in  $\marked(c) \setminus \Vset(\rho)$. 
{By \ref{P:yucktwo}, $u\notin \VI(\rho)$.} 
Since it is in $\marked(c)$, $u$  is a good variable.
We refer to $c$ as the \emph{``first clause of $\rho$''}
and $u$ as the \emph{``first variable of $\rho$''}.

We now define the four children of $\rho$ in the coupling tree.
For each of the four pairs~$(\tau_1,\tau_2)$ where $\tau_1$ and $\tau_2$
are assignments from $\{u\}$ to $\{\true,\false\}$, 
we create a child $\rhonext$ of~$\rho$ using Algorithm~\ref{alg:children}.
The following lemma shows that the coupling tree properties are maintained.

\begin{algorithm}[htbp]
  \caption{Create the child $\rhonext$}
  \label{alg:children}
    \begin{algorithmic}[1]
    \State $\Vset(\rhonext)  \gets \Vset(\rho) \cup \{u\}$; \label{stepVset}
    \State
      $\Aone(\rhonext) \gets $ combine $\Aone(\rho)$ with $\tau_1$;
      \State
      $\Atwo(\rhonext) \gets$ combine $\Atwo(\rho)$ with $\tau_2$;
    \State 
    $(\VI,\Crem,\failedclauses,\reason) \gets (\VI(\rho),\Crem(\rho),\failedclauses(\rho),\reason(\rho))$;
    \If {$\tau_1(u)\neq \tau_2(u)$}\label{step:equality-check}
      \State $\VI \gets \VI \cup \{u\}$;\label{step:VI-add-u}
      \For {$c': u\in\var(c')$}
        \State $\failedclauses \gets \failedclauses \cup \{c'\}$; \label{step:failed-disagree}
        \State $\reason(c') \gets \reason(c') \cup \{\disagree\}$; \label{step:today}
      \EndFor
    \EndIf
     \For{\label{step:satisfied}$c'\in \Crem$ \textbf{s.t.}\ $c'$ is satisfied by both $\Aone(\rhonext)$ and $\Atwo(\rhonext)$}
      \State $\Crem \gets \Crem \setminus \{c'\}$;\label{step:remove-sat}
    \EndFor 
\While{ $\exists c'\in \Crem$ 
with 
{$\var(c') \cap \VI \neq \emptyset$, $\var(c') \cap (\Va^{\partialVal}\setminus \VI) \neq \emptyset$, and $\marked(c') \setminus \Vset(\rhonext) = \emptyset$}
} 
\label{step:fail}
   \State  
      \texttt{(* For some $i\in\{1,2\}$, $c'$ is not satisfied by $\Ai(\rhonext)$. Otherwise,  $c'$  }
   \State     \hspace{1ex}    \texttt{ would have  been removed from $\Crem$ in the for loop of Line \ref{step:satisfied}. \hfill*)}
      \State $\failedclauses \gets \failedclauses \cup \{c'\}$; \label{step:failed-dissatisfy}
      \State $\reason(c') \gets \reason(c') \cup \{i\mid \text{$c'$ is not satisfied by $\Ai(\rhonext)$}\}$;

      \State $\VI \gets \VI \cup (\var(c')\setminus\Vset(\rhonext))$;\label{step:VI-add-fail}
               \State $\Crem \gets \Crem \setminus \{c'\}$;\label{step:remove-fail}
      \EndWhile

        \While{$\exists c'\in\Crem\cap\Cl_\bad$ with $\var(c')\cap\VI\neq\emptyset$}\label{step:bad}
      \State $\failedclauses \gets \failedclauses \cup \{c'\}$; \label{step:failed-bad}
\State  $\VI \gets \VI \cup (\var(c')\setminus\Vset(\rhonext))$;
\label{step:VI-add-bad}
      \State $\reason(c') \gets \reason(c') \cup \{\badsf\}$;
      \State $\Crem \gets \Crem \setminus \{c'\}$;\label{step:remove-bad}
    \EndWhile

\If{ $\exists c'\in \Crem$ 
with 
{$\var(c') \cap \VI \neq \emptyset$, $\var(c') \cap (\Va^{\partialVal}\setminus \VI) \neq \emptyset$, and $\marked(c') \setminus \Vset(\rhonext) = \emptyset$}
}   {Goto Line 15}    
\EndIf
        \State 
       $  (\VI(\rhonext),\Crem(\rhonext),\failedclauses(\rhonext),\reason(\rhonext))        \gets (\VI,\Crem,\failedclauses,\reason) $;     
       
  \end{algorithmic}
\end{algorithm}

\begin{lemma}
  If $\rho$ satisfies the coupling tree properties, then so does $\rhonext$. 
\end{lemma}

\begin{proof}
Property~\ref{P:xinVset} holds trivially for all nodes of the coupling tree.
Property~\ref{P:uselessclausesgone} holds because clauses  
 not satisfying the conditions 
are removed in the while loop of Line~\ref{step:fail}
{which is repeated after the loop in Line~\ref{step:bad} until there are no further changes.
}
Property~\ref{P:badclause} holds because clauses not satisfying the conditions are removed in the while loop of Line~\ref{step:bad}.
  For Property~\ref{P:rhoinv}, if  a clause is removed from $\Crem$, then it is removed in Line~\ref{step:remove-sat}, \ref{step:remove-fail}, or \ref{step:remove-bad}.
  The first case satisfies the first condition of \ref{P:rhoinv}, and the other two satisfy the second condition of \ref{P:rhoinv}.
  For Property~\ref{P:connected}, since $H_\Phi[\VI(\rho)]$ is connected, we just need to verify that when $\VI$ expands,
  the new vertices are connected with the old $\VI$.
  The set $\VI$ expands in Line~\ref{step:VI-add-u}   and Line~\ref{step:VI-add-fail} and \ref{step:VI-add-bad}.
  All cases can be verified straightforwardly.
  Property~\ref{P:VIvar} holds because all variables are added to $\VI$ in Line~\ref{step:VI-add-u}   and Line~\ref{step:VI-add-fail} and \ref{step:VI-add-bad}.
  In all three cases the corresponding clauses are added to $\failedclauses$.
  For Property~\ref{P:agree},  
  $\Vset$ is only expanded in Line~\ref{stepVset} and the property is guaranteed by Line~\ref{step:VI-add-u}.
    Property~\ref{P:failempty}  follows from  the way that $\reason$ and $\failedclauses$ are updated by the algorithm.

Property~\ref{P:typebad} follows from the fact that the loop in Line~\ref{step:bad} only applies to bad clauses. To see Property~\ref{P:typedisagree},
assume first that $c\in \failedclauses$ has $\disagree\in \reason(c)$.
This must happen in Line~\ref{step:today} and the loop guarantees 
that there is a variable  
$u \in \var(c) \cap \VI\cap \Vset$  
such that $\Aone(u) \neq \Atwo(u)$.  
For the other direction 
consider a good  clause~$c'$  that 
is added to~$\failedclauses$ in Step~\ref{step:failed-dissatisfy}.
All of the marked variables in~$c'$ are in $\Vset(\rhonext)$
so if there were a variable 
$v \in \var(c') \cap \Vset(\rhonext)$  
such that $\Aone(\rhonext)$ and $\Atwo(\rhonext)$ disagree on~$v$,
then  at the time that $v$ was set (either in $\rhonext$ or in a parent)
$\disagree$ would have been added to $\reason(c')$.
Property~\ref{P:typetwoclauses} follows by considering the loop in Line~\ref{step:fail}.

 {For Property~\ref{P:yucktwo} note that marked vertices can only be added to $\VI$  in Lines~\ref{step:VI-add-u}   and Line~\ref{step:VI-add-fail}. The vertex added in Line~\ref{step:VI-add-u}
is also in $\Vset$. The marked vertices that are added in Line~\ref{step:VI-add-fail}
are also all in $\Vset$.}
  \end{proof}

\subsection{Key property of the coupling tree for a random formula}

The following   property will be useful.
 Its proof is deferred to Section~\ref{sec:coupling-tree-property}.

\newcommand{\statelemdepthoftree}{\Whp over the choice of $\Phi$, for every prefix $\partialVal$ of $\partialVal^*$, every node~$\rho$ in $\=T^\partialVal$ has the property that
  $|\Vset(\rho)| \leq 3k^3\alpha L+1$.}
\begin{lemma}\label{lem:depthoftree}
\statelemdepthoftree
\end{lemma}

\begin{remark*}
  The number of nodes in the coupling tree is a polynomial in~$n$ since the depth of the tree 
  does not exceed~$\max_{\rho\in \=T^\partialVal}\abs{\Vset(\rho)}\le 
  3k^3\alpha L+1    = O(\log \frac{n}{\epsilon})$ and each node has at most $4$ children.
\end{remark*}

\section{The linear program}\label{sec:LP}

Before introducing the linear program 
which we will use to estimate
 $|\Omega_1^\partialVal|/|\Omega_2^\partialVal|$,
 we define one more piece of notation. In particular, for each node~$\rho$ of the coupling tree
 we define a quantity~$r(\rho)$ as follows.

Let $\rho$ be a node of the coupling tree.
Let $\CLI(\rho)$   be 
the set of clauses~$c\in\Cl^{\partialVal}$ such that $\var(c) \subseteq \VI(\rho)\cup\Vset(\rho)$.
For $i\in\{1,2\}$, let $N_i(\rho)$ be the number of assignments~$\tau$ to $\VI(\rho) \setminus \Vset(\rho) $ 
such that every clause in  
$\CLI(\rho)$ is satisfied by $\tau \cup \Ai(\rho)$.  

We will use the following lemma.

\begin{lemma}\label{lem:nonzero}
  If $\rho$ is a node in the coupling tree, 
  then $N_i(\rho)\neq 0$ for any $i\in\{1,2\}$.
\end{lemma}

\begin{proof}
Since our coupling tree is based on the marking from Lemma~\ref{lem:marking},
 there is a partial assignment of bad variables that satisfies all $c\in \Cl_\bad$. Let $\tau_\bad$ be such an assignment and let~$\tau_\bad(\rho)$ be the restriction of $\tau_\bad$ to the set 
 {$(\VI(\rho)\setminus \Vset(\rho))\cap \Va_\bad^\partialVal$.}
  Note that 
 {$\Vset(\rho) \subseteq \Va_\good^\partialVal$,} 
 so 
 (by Observation~\ref{obs:goodbad}) any bad clause $c\in \CLI(\rho)$ has the property that $\var(c) \subseteq (\VI(\rho)\setminus\Vset(\rho))\cap \Va_\bad$, which implies that $\tau_\bad(\rho)$ satisfies all bad clauses in $\CLI(\rho)$.
    
Next we claim that there is an partial assignment $\tau_\good(\rho):(\VI(\rho)\setminus \Vset(\rho))\cap \Va_\good^\partialVal\to \{\true, \false\}$ such that $\tau_\good(\rho)$ satisfies all good clauses in $\CLI(\rho)$. Let $c$ be a good clause in $\CLI(\rho)$. Again by 
Lemma~\ref{lem:marking}, $c$ has at least $k/4$ unmarked good variables. Note that $\var(c)\subseteq \VI(\rho)\cup \Vset(\rho)$ and $\Vset(\rho)$ consists only of marked variables, so $\abs{\var(c) \cap (\VI(\rho)\setminus \Vset(\rho))\cap \Va_\good^\partialVal}\geq k/4$. Denote by $A_c$ the set of assignments in $\Omega^*$ for which the restriction on $(\VI(\rho)\setminus \Vset(\rho))\cap\Va_\good^\partialVal$ does not satisfy $c$. Thus, $\Pr_{\mu_{\Omega^*}}(A_c) \leq 2^{-k/4}$. Also, by the definition of $A_c$, we obtain that $\var(A_c)\subseteq \var(c)$ and $\var(A_c) \subseteq \Va_\good^\partialVal$. So $\Gamma(A_c) \subseteq \Gamma_\good(c)$. Let $x(c) = \frac{1}{k\Delta}$. Since $\Gamma_\good(c)$ is the set of neighbours of $c$ in $G_{\Phi,\good}$, and the maximum degree in $G_{\Phi,\good}$ is at most $k(\HighD-1)$, we conclude our claim by applying Theorem~\ref{thm:lll} and verifying 
    \[
    x(c)\prod_{b\in \Gamma(A_c)}\Bigl(1-x(b)\Bigr) \ge \frac{1}{k\HighD}\left(1-\frac{1}{k\HighD}\right)^{k\HighD} \ge \frac{1}{e^2 k\HighD} > 2^{-k/4} \ge \Pr_{\mu_{\Omega^*}}(A_c).
    \]
    
    Now let $\tau = \tau_\good(\rho) \cup\tau_\bad(\rho)$. Then every clause in $\CLI(\rho)$ is satisfied by $\tau$, which yields that $N_i(\rho) > 0$.
\end{proof}

Define $r(\rho) := N_1(\rho)/N_2(\rho)$.
Lemma~\ref{lem:nonzero} implies that $r(\rho)$ is always well-defined.

\begin{observation}\label{obs:rho}    
  The quantity $r(\rho)$ can be computed in $m\cdot 2^{O(| \VI(\rho) \setminus \Vset(\rho)|)}$ time by considering all assignments $\tau$ to $ \VI(\rho) \setminus \Vset(\rho)$. 
 If $\rho$ is a leaf then $|\VI(\rho)| \leq L$, so   the time taken is polynomial in $n/\epsilon$.
\end{observation}

The importance of $r(\rho)$ comes from the following lemma.

\begin{lemma} \label{lem:rho}
  If $\rho$ is a leaf,  
  then $r(\rho) = |\Omega^{\Aone(\rho)\cup \Lambda}|/|\Omega^{\Atwo(\rho)\cup \Lambda}|$.
\end{lemma}

\begin{proof} 
Suppose $\rhois$.
Since $\rho$ is a leaf, by Definition~\ref{def:leaf},
for any $c\in\Crem$, either $\var(c)\subseteq \VI\cup \Vset$ (i.e., $c\in \CLI(\rho)$), or 
 {$\var(c)\subseteq \Va^\partialVal\setminus \VI$ (or both).}
Denote by $\CLO(\rho)$ the set of clauses in 
  {$\Crem\setminus \CLI(\rho)$} 
such that
  {$\var(c)  \subseteq \Va^\partialVal\setminus \VI$.}
Denote by $\Cother(\rho)$ the set of clauses in 
$\Cl^{\partialVal} \setminus (\Crem \cup \CLI(\rho))$.
Thus, the clauses in $\Cl^{\partialVal}$ split into  the disjoint sets $\CLI(\rho)$, $\CLO(\rho)$ and $\Cother(\rho)$.

By \ref{P:rhoinv}, every clause in $\Cother(\rho)$
is satisfied by $\Aone(\rho)$ and $\Atwo(\rho)$.

Let $M(\rho)$ be the number of assignments $\sigma \colon \Va^\partialVal\setminus(\VI\cup\Vset) \to \{\true,\false\}$
  such that all clauses in $\CLO$ are 
 {satisfied by $\sigma\cup \Aone(\rho)$
or equivalently, by
\ref{P:agree}, 
satisfied by $\sigma\cup \Atwo(\rho)$.}
  Recall that
$N_i(\rho)$ is the number of assignments~$\tau$ to $\VI(\rho) \setminus \Vset(\rho) $ 
such that every clause in  
$\CLI(\rho)$ is satisfied by $\tau \cup \Ai(\rho)$.  
Then for $i\in\{1,2\}$, $|\Omega^{\Ai(\rho)\cup \Lambda}|=M(\rho)N_i(\rho)$, which implies the lemma, 
since $r(\rho) = N_1(\rho)/N_2(\rho)$.\end{proof}

We will use a binary search to approximate the quantity $\abs{\Omega^\partialVal_1}/\abs{\Omega^\partialVal_2}$.
Our linear program  relies on two constants $\rlower$ and $\rupper$.
 We will move these closer and closer together by binary search. For each
  node~$\rho$ of the coupling tree, 
we introduce two variables $P_{1,\rho}$ and $P_{2,\rho}$.
The idea is  that a solution of the LP should have the property that
$$\frac{\abs{\Omega^\partialVal_1}}{\abs{\Omega^\partialVal_2}} = \frac{P_{1,\rho} }{P_{2,\rho}  }\cdot \frac{\abs{\Omega^{\Aone(\rho)\cup \Lambda}}}{\abs{\Omega^{\Atwo(\rho)\cup \Lambda}}}.$$

We now introduce the constraint sets of the LP.

\subsection*{Constraint Set 0}

For every node $\rho$ of the coupling tree 
and every $i\in \{1,2\}$ we add the constraint
$ 0 \leq P_{i,\rho} \leq 1$.

\subsection*{Constraint Set 1}

If  $\rho\in\calL$ then we add the following constraints.
\begin{align*}
  \rlower \,P_{2,\rho}  & \leq P_{1,\rho}\,   r(\rho)\\
  P_{1,\rho} \, r(\rho) &\leq \rupper \,P_{2,\rho}
\end{align*}

\begin{remark*}
  The purpose of these constraints is to guarantee
  $$\rlower\leq  \frac{P_{1,\rho} }{P_{2,\rho}  }\, r(\rho)\leq \rupper. $$ 
\end{remark*}

\subsection*{Constraint Set 2}

For the root $\rho^*$ of the coupling tree, we add the constraints
$$P_{1,\rho^*} = P_{2,\rho^*} = 1.$$

For every node $\rho$ of the coupling tree that is not in $\calL^*$, 
let $u$ be the first variable of~$\rho$.  
Add constraints as follows.
For each $X\in \{\true,\false\}$ add the following constraints.
\begin{align*}
  P_{1,\rho} &= P_{1,\rho_{\sigX,\sigtrue}} + P_{1,\rho_{\sigX,\sigfalse}}\\ 
  P_{2,\rho} &= P_{2,\rho_{\sigtrue,\sigX}} + P_{2,\rho_{\sigfalse,\sigX}}
\end{align*}

These   constraints imply the following lemma.

\begin{lemma}\label{lem:constraints2}
  Suppose that the LP variables satisfy all of the constraints in Constraint Set 2.
  Then for any $i\in \{1,2\}$ and
  any $\sigma\in \Omega_i^\partialVal$,
  $$\sum_{\rho\in \calL^*:\sigma  \in \Omega^{\Ai(\rho)\cup \Lambda}} P_{i,\rho} = 1$$
\end{lemma}
\begin{proof}
  For $i\in\{1,2\}$, we will maintain a set $\Psi_i$ of nodes in the coupling tree with the 
  invariant that $\sum_{\rho \in \Psi_i} P_{i,\rho} = 1$
  and every node $\rho \in \Psi_i$ has $\Ai(\rho)$ agree with~$\sigma$.
    For the base case, we let $\rho^*$ be the root of the coupling tree
  and we take $\Psi_i = \{\rho^*\}$. If every node in $\Psi_i$ is in $\calL^*$ then we are finished.
  Otherwise, let $\rho$ be some node in $\Psi_i$ that is not in $\calL^*$.
  Let $u$ be the first variable of~$\rho$.
  Let $\rho'$ and $\rho''$ be the two children of~$\rho$ 
  such that $\Ai(\rho')$ and $\Ai(\rho'')$ both map $u$ to $\sigma(u)$.
  Replace $\rho$ in $\Psi_i$ with $\rho'$ and $\rho''$.
  The constraints guarantee that $P_{i,\rho} = P_{i,\rho'} + P_{i,\rho''}$, so the invariant is maintained.
\end{proof}

\subsection*{Constraint Set 3} 

For every node $\rho$ of the coupling tree that is not in~$\calL^*$, 
every $X\in \{\true,\false\}$, and every $i\in \{1,2\}$,
let $u$ be the first variable of~$\rho$
and add the constraint $P_{i,\rho_{\sigX,\sigXbar}} \leq \frac{1}{s}\, P_{i,\rho}$. 

The intuition behind this set of constraints is Lemma~\ref{lem:local-uniform}. 

\section{Analysis of the linear program}

In Section~\ref{sec:completeness}, we  show that whenever $\rlower\leq {\abs{\Omega^\partialVal_1}}/{\abs{\Omega^\partialVal_2}} \leq \rupper$,
a solution to the LP exists.
 We call this ``completeness'' of the LP.  
In the remaining subsections of this section, we show  ``soundness'' --- namely, that whenever a solution exists, 
$\rlower$ and $\rupper$ are valid bounds for the quantity ${\abs{\Omega^\partialVal_1}}/{\abs{\Omega^\partialVal_2}}$, up to small errors.

\subsection{Completeness}\label{sec:completeness}

Recall that we  use the marking from Lemma~\ref{lem:marking}
and that $\partialVal$ is a prefix of the partial assignment $\partialVal^*$ from Lemma~\ref{lem:partial-assignment}.

We use the following lemma. 
\begin{lemma}
\label{lem:anothernonzero}
If $\rho$ is a node in the coupling tree then, for any $i\in \{1,2\}$, $\Omega^{\Ai(\rho)\cup \Lambda}$
is non-empty.
\end{lemma}
\begin{proof}
Let $\rho$ be a node in the coupling tree and fix $i\in \{1,2\}$.
Since our coupling tree is based on the marking from Lemma~\ref{lem:marking},
 there is a partial assignment of bad variables that satisfies all $c\in \Cl_\bad$. Let $\tau_\bad$ be such an assignment.
 Note that $\Ai(\rho) \cup \Lambda$ is a partial assignment of marked (good) variables 
 so it does not assign any variables in common with $\tau_\bad$.
 Let $V'$ be the set of unmarked good variables.
 It remains to show that there is a partial assignment of variables in~$V'$ 
  that satisfies all clauses in $\Cl_\good$.
  To do this we apply Theorem~\ref{thm:lll} to $\Va_\good$ and $\Cl_\good$ in the same way as the proof of Lemma~\ref{lem:local-uniform}. 
\end{proof}

Let $\rho$ be a node of the coupling tree with first variable~$u$.
For $X\in \{\true,\false\}$, we use the notation 
\begin{align}\label{eqn:psi}
  \psi_{\rho,X,1} := \frac{|\Omega^{\Aone(\rho_{\sigX,\sigX} )\cup \Lambda}|}{|\Omega^{\Aone(\rho)\cup \Lambda}|} 
  = \frac{|\Omega^{\Aone(\rho_{\sigX,\sigXbar}) \cup \Lambda}|}{|\Omega^{\Aone(\rho)\cup \Lambda}|}.
\end{align}
This is well-defined since $\Aone(\rho_{\sigX,\sigX}) = \Aone(\rho_{\sigX,\sigXbar})$.
In other words, $\psi_{\rho,X,1}$ is the probability that~$u$ is assigned value~$X$ under $\mu_{\Omega^{\Aone(\rho)\cup \Lambda}}$.
Thus, $\psi_{\rho,X,1} + \psi_{\rho,\neg X,1}=1$. 
We 
similarly define
$$
  \psi_{\rho,X,2} := \frac{|\Omega^{\Atwo(\rho_{\sigX,\sigX} )\cup \Lambda}|}{|\Omega^{\Atwo(\rho)\cup \Lambda}|} 
  = \frac{|\Omega^{\Atwo(\rho_{\sigXbar,\sigX}) \cup \Lambda}|}{|\Omega^{\Atwo(\rho)\cup \Lambda}|},
$$
 noting that $\psi_{\rho,X,2} + \psi_{\rho,\neg X,2}=1$. 

We will next give an inductive definition of a function  $\pseudoProb$ from nodes of the coupling tree to real numbers in $[0,1]$.
The way to think about this is as follows --- we will implicitly define a probability distribution over paths from the root of the coupling
tree to $\calL^*$. For each node~$\rho$, $Q(\rho)$ will be the probability that $\rho$ is included in a path drawn from this distribution.

Any such path starts at the root, so we define
  $\pseudoProb (\rho^*)=1$. 
Once we have defined $\pseudoProb (\rho)$ for a node $\rho$ that is not in $\calL^*$
we can define $\pseudoProb (\cdot)$ on the children of~$\rho$ as follows.
Let $u$ be the first variable of $\rho$ and consider the four children $\rho_{\sigtrue,\sigtrue},\rho_{\sigtrue,\sigfalse},\rho_{\sigfalse,\sigtrue},\rho_{\sigfalse,\sigfalse}$.
Define the values of~$Q$ as follows.
\begin{align}
  \begin{split}\label{eqn:Q}
    \pseudoProb (\rho_{\sigtrue,\sigtrue}) &:= \pseudoProb (\rho)  \min\{ \psi_{\rho,\true,1},\psi_{\rho,\true,2}\}.\\
    \pseudoProb (\rho_{\sigtrue,\sigfalse}) &:= \pseudoProb (\rho)  (\psi_{\rho,\true,1} - \min\{ \psi_{\rho,\true,1},\psi_{\rho,\true,2}\}).\\
    \pseudoProb (\rho_{\sigfalse,\sigfalse}) &:= \pseudoProb (\rho) \min\{1-\psi_{\rho,\true,1},1-\psi_{\rho,\true,2} \}.\\
    \pseudoProb (\rho_{\sigfalse,\sigtrue}) &:= \pseudoProb (\rho) ((1-\psi_{\rho,\true,1}) -  \min\{1-\psi_{\rho,\true,1},1-\psi_{\rho,\true,2} \} ).  
  \end{split}
\end{align}

\begin{observation}\label{obs:couple}
For any  $X\in \{\true,\false\}$,
  \begin{align*}
    \pseudoProb (\rho_{\sigX,\sigtrue}) + \pseudoProb (\rho_{\sigX,\sigfalse}) &= \pseudoProb (\rho) \psi_{\rho,X,1}, and\\
    \pseudoProb (\rho_{\sigtrue,\sigX}) + \pseudoProb (\rho_{\sigfalse,\sigX}) &= \pseudoProb (\rho) \psi_{\rho,X,2}.
  \end{align*}
  Also,
  $$
  \pseudoProb (\rho_{\sigtrue,\sigtrue}) + \pseudoProb (\rho_{\sigtrue,\sigfalse}) + \pseudoProb (\rho_{\sigfalse,\sigtrue}) + \pseudoProb (\rho_{\sigfalse,\sigfalse}) = \pseudoProb (\rho).
$$ 
   This implies  that $\sum_{\rho \in \calL^*} \pseudoProb (\rho) = 1$.
\end{observation}

We use $\pseudoProb$ to define a feasible solution to the LP.
Indeed, this definition  explains what the variables in the LP were meant to represent.

\begin{definition}[The LP variables] \label{def:LPvar}
  For each node $\rho$ of the coupling tree and each $i\in \{1,2\}$, 
  define  
  $P_{i,\rho}:= \pseudoProb (\rho)  {|\Omega^\partialVal_{i}|}/{|\Omega^{\Ai(\rho)\cup \Lambda}|}$.
\end{definition}

Theorem~\ref{lem:anothernonzero} ensures the values in Definition~\ref{def:LPvar} are well defined.
 Motivated by Definition~\ref{def:LPvar}, we give the following upper bound for~$Q(\rho)$.

\begin{lemma}\label{lem:getcons0}
  For all nodes $\rho$ in the coupling tree
  and all $i\in \{1,2\}$, $\pseudoProb (\rho) \leq  {|\Omega^{\Ai(\rho)\cup \Lambda}|}/{|\Omega^\partialVal_i|}$.
\end{lemma}
\begin{proof}
  The proof is by induction --- having established the claim for a node~$\rho$, we then establish it for the children of~$\rho$
  using Observation~\ref{obs:couple}.
  For the base case, 
  $\pseudoProb (\rho^*) = 1 =  |\Omega^{\Ai(\rho^*)\cup \Lambda}|/|\Omega^\partialVal_i|$.
  For the inductive step, consider a node~$\rho$ (for which the claim is established) and let $u$ be the first variable of $\rho$.
  Consider any child $\rho_{\sigX,\sigY}$ of~$\rho$. We show the lemma for $i=1$ and the other case is similar.
  By Observation~\ref{obs:couple}, Equation~\eqref{eqn:psi}, and the induction hypothesis,
  \begin{align*}
    \pseudoProb (\rho_{\sigX,\sigY}) &\le \pseudoProb (\rho)\cdot \psi_{\rho,X,1}= \pseudoProb (\rho)\cdot\frac{\abs{\Omega^{\Aone(\rho_{\sigX,\sigY})\cup \Lambda}}}{\abs{\Omega^{\Aone(\rho)\cup \Lambda}}}\\
    &\le \frac{\abs{\Omega^{\Aone(\rho)\cup \Lambda}}}{\abs{\Omega^\partialVal_i}}\cdot \frac{\abs{\Omega^{\Aone(\rho_{\sigX,\sigY})\cup \Lambda}}}{\abs{\Omega^{\Aone(\rho)\cup \Lambda}}}
    =\frac{\abs{\Omega^{\Aone(\rho_{\sigX,\sigY})\cup \Lambda}}}{\abs{\Omega^\partialVal_i}}.\qedhere
  \end{align*}
\end{proof}

The next lemma relates the values of the LP variables, as defined by Definition~\ref{def:LPvar}.

\begin{lemma}\label{lem:docouple}
  Let $\rho\not\in\+L^*$ be a node in the coupling tree. 
  Let $u$ be the first variable of~$\rho$.
  For any $X\in \{\true,\false\}$ and $i\in \{1,2\}$,
 $   P_{i,\rho_{\sigX,\sigXbar}} \leq    P_{i,\rho}/s$.

\end{lemma}

\begin{proof}
  We first assume $X=\true$ and $i=1$.
  Definition~\ref{def:LPvar} implies that  
  \begin{align} \label{eqn:P-ratio}
    \frac{P_{i,\rho_{\sigtrue,\sigfalse}}}{P_{i,\rho}} = \frac{Q(\rho_{\sigtrue,\sigfalse})}{Q(\rho)}\,\frac{\abs{\Omega^{\Aone(\rho)\cup \Lambda}}}{ \abs{\Omega^{\Aone( \rho_{\sigtrue,\sigfalse})\cup \Lambda}}}.
  \end{align}
  From the definition \eqref{eqn:Q} 
  (for the case where $\psi_{\rho,T,2} \leq \psi_{\rho,T,1}$)
  along with Observation~\ref{obs:couple} (for the other case), we have that
  \begin{align*}
    \frac{\pseudoProb (\rho_{\sigtrue,\sigfalse})}{\pseudoProb (\rho)} \leq |\psi_{\rho,\true,1} - \psi_{\rho,\true,2}|.
  \end{align*}
  Recall that $\psi_{\rho,\true,i}$ is the probability that $u$ is assigned the value~$\true$ under $\mu_{\Omega^{\Ai(\rho)\cup \Lambda}}$.
  By Lemma~\ref{lem:anothernonzero}, $\Omega^{\Ai(\rho)\cup \Lambda}$
  is non empty, so we can apply Lemma~\ref{lem:local-uniform} to this partial assignment.  From the second part of the lemma, we have
  \begin{align*}
    \frac{1}{2}\left( 1-\frac{1}{3s} \right) \leq \psi_{\rho,\true,i} \leq \frac{1}{2}\left( 1+\frac{1}{3s} \right).
  \end{align*}
  Since $k$ is sufficiently large and so is $s$,
  \begin{align*}
    |\psi_{\rho,\true,1} - \psi_{\rho,\true,2}| \leq \frac{1}{s}\cdot \psi_{\rho,\true,i}.
  \end{align*}
  The claim in the lemma follows 
  for $X=\true$ and $i=1$
  since $\psi_{\rho,\true,1} = \frac{ \abs{\Omega^{\Aone( \rho_{\sigtrue,\sigfalse})\cup \Lambda}}}{\abs{\Omega^{\Aone(\rho)\cup \Lambda}}}$.

  For $X=\false$ and $i=1$,  note that
  \begin{align*}
    \frac{\pseudoProb (\rho_{\sigfalse,\sigtrue})}{\pseudoProb (\rho)} \leq | (1 - \psi_{\rho,\true,1}) - (1-\psi_{\rho,\true,2})| =   |\psi_{\rho,\true,1} - \psi_{\rho,\true,2}|,
  \end{align*}
  and use $\psi_{\rho,\false,1} = \frac{ \abs{\Omega^{\Aone( \rho_{\sigfalse,\sigtrue})\cup \Lambda}}}{\abs{\Omega^{\Aone(\rho)\cup \Lambda}}}$ in the end. 
  
  For $i=2$, the proof is similar.
\end{proof}

Now we are ready to show the completeness.

\begin{lemma}\label{lem:existence}
  Suppose $\rlower \leq  {|\Omega^\partialVal_1|}/{|\Omega^\partialVal_2|} \leq \rupper$.
  The variables $\*P=\{P_{i,\rho}\}$ defined in Definition~\ref{def:LPvar} satisfy all constraints of the LP.
\end{lemma}

\begin{proof}
  \ConZero:\quad 
  The fact that the LP variables satisfy these constraints follows from the definition of the variables in Definition~\ref{def:LPvar}
  (which guarantees that they are all non-negative) and from Lemma~\ref{lem:getcons0}.
  
  \ConOne:\quad Definition~\ref{def:LPvar} implies that for any node~$\rho$ in the coupling tree,
  \begin{align}\label{eq:LPgoal}
    \frac{P_{1,\rho}}{P_{2,\rho}} \, \frac{|\Omega^{\Aone(\rho)\cup \Lambda}|}{|\Omega^{\Atwo(\rho)\cup \Lambda}|} 
    = \frac{|\Omega^\partialVal_1|}{|\Omega^\partialVal_2|}.
  \end{align}
  
  If $\rho$ is a leaf, then  by Lemma~\ref{lem:rho}, $r(\rho) = {|\Omega^{\Aone(\rho)\cup \Lambda}|}/{|\Omega^{\Atwo(\rho)\cup \Lambda}|}$.
  Equation~\eqref{eq:LPgoal}
  implies that
  $$\frac{P_{1,\rho}}{P_{2,\rho}} \, r(\rho) = \frac{|\Omega^\partialVal_1|}{|\Omega^\partialVal_2|},$$
  so as long as 
  $\rlower \leq  {|\Omega^\partialVal_1|}/{|\Omega^\partialVal_2|} \leq \rupper$,
  the LP variables satisfy the constraints in \ConOne, as required.

  \ConTwo:\quad 
  For the root $\rho^*$ of the coupling tree, 
  it is easy to see from Definition~\ref{def:LPvar} that $P_{i,\rho^*} = 1$.

  Let $\rho\not\in\+L^*$ be a node in the coupling tree and let $u$ be the first variable of~$\rho$.
  For $X\in \{\true,\false\}$ and $i=1$ we wish to establish 
  $P_{1,\rho} = P_{1,\rho_{\sigX,\sigtrue}} + P_{1,\rho_{\sigX,\sigfalse}}$.
  Plugging in Definition~\ref{def:LPvar} and dividing by $|\Omega^\partialVal_1|$, 
  the constraint is equivalent  (for any $Y\in \{\true,\false\}$) to
  \begin{align*}
    \frac{\pseudoProb(\rho)  }{|\Omega^{\Aone(\rho)\cup \Lambda}|}
    = \frac{\pseudoProb(\rho_{\sigX,\sigtrue})}{|\Omega^{\Aone( \rho_{\sigX,\sigtrue})\cup \Lambda}|} + 
    \frac{\pseudoProb(\rho_{\sigX,\sigfalse})}{|\Omega^{\Aone( \rho_{\sigX,\sigfalse})\cup \Lambda}|} 
    = \frac{\pseudoProb(\rho_{\sigX,\sigtrue})}{|\Omega^{\Aone( \rho_{\sigX,\sigY})\cup \Lambda}|} + 
    \frac{\pseudoProb(\rho_{\sigX,\sigfalse})}{|\Omega^{\Aone( \rho_{\sigX,\sigY})\cup \Lambda}|},
  \end{align*}
  where we used again the fact that $\Aone(\rho_{\sigX,\sigX}) = \Aone(\rho_{\sigX,\sigXbar})$.
  The equation above follows from Observation~\ref{obs:couple} using~\eqref{eqn:psi}. 
  The other three constraints are similar.

  \ConThree:\quad This case directly follows from Lemma~\ref{lem:docouple}.
\end{proof}

\subsection{\texorpdfstring{$\ell$}{l}-wrong assignments}

There are two kinds of errors which cause solutions of the LP to differ from the ratio ${\abs{\Omega^\partialVal_1}}/{\abs{\Omega^\partialVal_2}}$.
The first kind of error involves a notion that we call 
  ``$\ell$-wrong assignments''.
To define them, we need some graph-theoretic notation.

\begin{definition}\label{def:graph-notations}
  Given a graph $G$ with vertices~$u$ and~$v$ in $V(G)$,
  let $\dist_G(u, v)$ be the  distance between $u$ and $v$ in $G$ --- that is, the number of edges in a shortest path from~$u$ to~$v$. 
  Given a subset $T\subseteq V(G)$ and a vertex $v\in V(G)$,
  let $\dist_G(u,T)$ be $\min_{v\in T} \dist_G(u,v)$.
  For any positive integer~$k$, 
  let $G^k$ be the graph with vertex set~$V(G)$ in which vertices $u$ and $v$ are connected 
  if and only if there is a length-$k$ path from~$u$ to~$v$ in~$G$.
  Let $G^{\leq k}$ be the graph with vertex set~$V(G)$ 
  in which vertices $u$ and $v$ are connected if and only if 
  there is a path from~$u$ to~$v$ in~$G$ of length at most~$k$.
\end{definition}

The main combinatorial structure
that we use is a set $\calD(G_\Phi)$, which 
is based on Alon's ``2,3-tree''~\cite{Alon}. Similar structures were subsequently used in~\cite{Moi19,GLLZ19}.
The main difference between 
our definition and previous ones is that we 
take into account whether clauses are connected via good variables.

\begin{definition}\label{def:disjoint}
  Given the graph~$G_{\Phi}$,
  let $\calD(G_{\Phi})$ be the set of subsets $T\subseteq V(G_{\Phi})$ 
  such that the following hold:
  \begin{enumerate}
    \item For any $c_1,c_2\in T$, $\var(c_1) \cap \var(c_2)\cap \Va_\good = \emptyset$;\label{item:DG-disjoint}
    \item The graph $G_{\Phi}^{\leq 4}[T]$, which is the subgraph of $G_{\Phi}^{\leq 4}$ induced by~$T$, is connected.
  \end{enumerate}
\end{definition}

The following two lemmas regarding $\calD(G_{\Phi})$ will be useful.
We defer their proofs to Section~\ref{sec:DG}.

\newcommand{\statelemnumtrees}{Let $\ell$ be an integer which is at least $\log n$.
  \Whp over the choice of $\Phi$, 
  every clause  $c\in\Cl^\partialVal_\good$ has the property that
  the number of size-$\ell$ subsets $T \in \calD(G_{\Phi})$   containing $c$ is at most $(18k^2\alpha)^{4\ell}$.}
\begin{lemma}\label{lem:numtrees}
\statelemnumtrees
\end{lemma}

Recall that the coupling tree~$\=T^\partialVal$ 
 is defined with respect to
 a prefix $\partialVal$ of the partial assignment $\partialVal^*$ from Lemma~\ref{lem:partial-assignment}.
Let $c^*$ be the first clause of the root node $\rho^*$.

\newcommand{\statelemlargetree}{\Whp over the choice of $\Phi$, every node $\rho$ in   $\=T^\partialVal$   
 with $|\VI(\rho)|\geq L$ 
  has the property  that    there  is a set $T \subseteq \failedclauses(\rho)$ containing $c^*$ such that  
    $T\in \calD(G_{\Phi})$,  $\abs{T}  = C_0\lceil\log(n/\eps)\rceil$ and $\abs{T\cap \Cl_\bad} \le  {\abs{T}}/{3}$.}
\begin{lemma}\label{lem:large3tree} 
\statelemlargetree
\end{lemma}

 We now define $\ell$-wrong assignments.

\begin{definition}\label{def:l-wrong}
    An assignment $\sigma \in \Omega_i^\partialVal$ is $\ell$-wrong if there is a size-$\ell$ set $T\in \calD(G_{\Phi})$ such that
  \begin{itemize}
  \item $c^*\in T$,
    \item $|T\cap\Cl^\partialVal_\good|\geq 2\abs{T}/3$, and
    \item there is a size $\lceil   \ell/2\rceil $ subset $S$ of~$T \cap \Cl^\partialVal_\good$ 
    such that  the restriction of~$\sigma$ to marked variables in clauses in~$S$ does not satisfy any clause in~$S$.
    (Formally, taking $U$ to be $\cup_{c\in S} \marked(c)$, the condition is that $\sigma[U]$ does not satisfy any clauses in~$S$.)
      \end{itemize}
  Otherwise $\sigma$ is $\ell$-correct. 
  \end{definition}

The following is similar to \cite[Lemma 4.8]{GLLZ19}.

\begin{lemma}\label{lem:15}
Fix $i\in \{1,2\}$.
  Let $\ell=L/(3k^2\Delta)$. Then the fraction of assignments in $\Omega_i^\partialVal$ that are $\ell$-wrong is at most $(k\HighD)^{-9\ell}$.
\end{lemma}
\begin{proof}
 Assume $i=1$, as the other case is similar.
  We want to show that
  \begin{align*}
    \Pr_{\sigma\sim\mu_{\Omega^\partialVal_1}}(\text{$\sigma$ is $\ell$-wrong}) \le (k\HighD)^{-9\ell}.
  \end{align*}
  If every assignment in $\mu_{\Omega^\partialVal_1}$ is $\ell$-correct, then we are finished.
  Otherwise,
  consider a size-$\ell$ set $T = \{c_1, c_2, \ldots, c_\ell\}$ 
   in $\calD(G_\Phi)$ such that $c_\ell = c^*$ and   $\ell_T \geq 2\ell/3$ where $\ell_T:=|{T\cap \Cl_\good^\partialVal}|$. 

  Let $V:=\Va^\partialVal\setminus\{v^*\}$ be the set of unassigned variables in~$\rho^*$,
  and let $\Omega'$ be the set of all assignments $\sigma:V\rightarrow\{\true,\false\}$.  
  Let $U_T:=V\cap\left(\cup_{i=1}^\ell \marked(c_i)\right)$ be the set of all   marked variables in~$V$ that are  in clauses in $T$. 
  Given an assignment $\sigma \in \Omega'$, let $Z_T(\sigma)$ be the number of clauses in $T\cap \Cl_\good^\partialVal$ that are not satisfied by $\sigma[U_T]$.

  Suppose that $\sigma$ is drawn from $\mu_{\Omega'}$.
  By Corollary~\ref{cor:partial-prefix}, every $c\in\Cl_\good^\partialVal$
  has at least $k/4-1$ marked variables in~$V$. 
  Thus, for every   $c\in T\cap \Cl_\good^\partialVal$, the probability that $c$  is not satisfied by $\sigma[U_T]$ is at most $p^*:=2^{1-k/4}$.
  For $c\in T\cap \Cl_\good^\partialVal$, let $A_c$ be the event that $c$ is not satisfied by $\sigma[U_T]$.   Because $T \in \calD(G_{\Phi})$,   clauses in $T$  do not share good variables (so they do not share marked variables).
  Thus,  the events $A_c$ are mutually independent. Since $\ell_T=|T\cap \Cl_\good^\partialVal|$, we have that the event $Z_T(\sigma)\geq \ceil{\ell/2}$ is dominated above by the probability that $\widehat{Z}\geq \ceil{\ell/2}$ where $\widehat{Z}$ is a sum of $\ell_T$ independent Bernoulli r.v.s with parameter $p^*$. Setting $\gamma := 
   \ell/(2p^* \ell_T)-1 \geq 1/(2p^*) -1$ and applying a Chernoff bound to $\widehat{Z}$, we obtain that   \begin{align*}
    \Pr_{\mu_{\Omega'}}(Z_T(\sigma)\geq  \ceil{\ell/2})& = \Pr_{\mu_{\Omega'}}(Z_T(\sigma)\geq \ell/2)=\Pr_{\mu_{\Omega'}}(Z_T(\sigma)\geq  (1+\gamma)p^*\ell_T)\\
		&\leq \Pr(\widehat{Z}\geq (1+\gamma)p^*\ell_T)\leq \biggl(\frac{e^\gamma}{(1+\gamma)^{1+\gamma}}\biggr)^{p^*\ell_T} \leq \Bigl( {2ep^*} \Bigr)^{\ell_T/2}\leq \Bigl({2ep^*}\Bigr)^{\ell/3},
  \end{align*} 
  where the second-to-last inequality follows by substituting in the lower bound $1/(2p^*)-1$ for $\gamma$.

  Let $\mathcal{E}_T$ be the event that $Z_T(\sigma)\geq\ceil{\ell/2}$, so $\var(\mathcal{E}_T) = U_T$.
  We apply Lemma~\ref{lem:local-uniform} 
  with the partial assignment $\partialVal\cup\{v^*\rightarrow\true\}$
  and with the event~$A$ from Lemma~\ref{lem:local-uniform}
  being~$\mathcal{E}_T$.
  A clause $c$ is in $\Gamma(\mathcal{E}_T)$ if and only if $\var(c)$ intersects $U_T$. 
  Since all variables in $U_T$ are marked,
  each $c_i\in T\cap \Cl_\good^\partialVal$ has at most $3k/4$ variables in $U_T$.
  Moreover, for each $v\in U_T$, there are at most $\HighD$ clauses containing $v$.
  Thus, $\abs{\Gamma(\mathcal{E}_T)}\le 3k\ell\HighD/4$.

  Plugging everything into Lemma~\ref{lem:local-uniform},  we obtain
  \begin{align*}
    \Pr_{\mu_{\Omega^\partialVal_1}}(\mathcal{E}_T) \leq \Pr_{\mu_{\Omega'}}(\mathcal{E}_T) \Bigl(1-\frac{1}{sk\HighD}\Bigr)^{- 3k\ell\HighD/4}\leq\Bigl(
    {2ep^*}\Bigr)^{\ell/3}  e^{\ell/s},
  \end{align*}
  where the value~$s$ (defined just before Lemma~\ref{lem:local-uniform}) is $s=2^{k/4}/(ek\HighD)$.
  
  Note from the definition of~$L$ that $\ell = C_0 \lceil \log(n/\epsilon) \rceil$ for a sufficiently large constant~$C_0$, so $\ell\geq \log n$.
 Applying  Lemma~\ref{lem:numtrees} to the good clause~$c^*$, \whp 
  over the choice of~$\Phi$, the number of size-$\ell$ sets $T\in \calD(G_\Phi)$  containing $c^*$ is at most $(18k^2\alpha)^{4\ell}$. 
  Hence, by a union bound,   \begin{align*}
    \Pr_{ \mu_{\Omega^\partialVal_1}}(\sigma \mbox{ is $\ell$-wrong}\,\!) 
    &\leq \sum_{T\in \calD(G_\Phi): \abs{T}=\ell, c^*\in T,
    \ell_T \geq 2 \ell/3
    } \Pr_{\mu_{\Omega^\partialVal_1}}(\mathcal{E}_T)\\
    &\leq (18k^2\alpha)^{4\ell}\biggl({2ep^*}\biggr)^{\ell/3} e^{\ell/s} = \left( 
    18^4 k^8 \alpha^4 2^{1/3} e^{1/3} 2^{1/3} 2^{-k/12} e^{1/s} \right)^\ell\\
    &\le (k\HighD)^{-9\ell}, 
  \end{align*}
where in the last inequality we used that $\alpha\leq 2^{k/100}$ and $\Delta=2^{k/300}$. 
\end{proof}

\subsection{Errors from truncating nodes}

We have already noted that the existence of $\ell$-wrong assignments
can be viewed as ``errors''   which cause solutions of the LP to differ from the ratio ${\abs{\Omega^\partialVal_1}}/{\abs{\Omega^\partialVal_2}}$.
However, even if an assignment  $\sigma\in \Omega_i^\partialVal$ is $\ell$-correct,
it may still induce some error in the solution of the LP due to the truncation of the coupling tree.
We account for this kind of error in this section.

Fix an assignment $\sigma\in \Omega_1^\partialVal$. 
Given a solution ${\*P} = \{P_{i,\rho}\}$ to the LP,  
we consider the following stochastic process, which is a process (depending on~$\sigma$) for  sampling   from 
 $\+L^*$.
 The process  
 is defined for $i=1$ rather than for $i=2$. That is it starts with $\sigma\in \Omega_1^\partialVal$ and it uses the LP values $P_{1,\rho}$.
 (We will later require a similar process for $i=2$).

 \begin{definition}[Sampling conditioned on $\sigma\in \Omega_1^\partialVal$]\label{def:thecoupling}
Fix $\sigma\in \Omega_1^\partialVal$.
Here is a method for choosing a node $\rho \in \calL^*$.
Start by setting $\rho$ to be the root $\rho^*$ of the coupling tree.
From every node $\rho$ that is not in $\calL^*$, proceed to a child, as follows. 
Let $u$ be the first variable of~$\rho$.
 For each $X\in\{\true,\false\}$, go to $\rho_{u\rightarrow\sigma(u),u\rightarrow X}$ with probability $ {P_{1,\rho_{u\rightarrow\sigma(u),u\rightarrow X}}}/{P_{1,\rho}}$.
\end{definition}

This process is well-defined because ${\*P}  $ satisfies \ConTwo.
 If $\rho$ is a leaf such that $\Aone(\rho)$ agrees with $\sigma$,
 then the probability of reaching $\rho$  is $P_{1,\rho}$.
The following definition  concerns the truncating nodes $\rho\in \mathcal{T}$ 
where the process can also stop.   

\begin{lemma}\label{lem:typetwo} 
  Fix $\rlower \leq \rupper$.
  Let $\ell= L/(3k^2\HighD)  $.
  \Whp over the choice of $\Phi$, the following holds.
   Let $\sigma$
  be any $\ell$-correct assignment in $  \Omega^\partialVal_1$.
  If the LP has a solution using $\rlower$ and $\rupper$ then
  $ \sum_{\rho\in \calT:\sigma \in \Omega^{\partialVal\cup\Aone(\rho)} } P_{1,\rho} \leq \left( k\HighD \right)^{-8\ell}$.
\end{lemma}

\begin{proof} 
Let $\pi_\sigma$ be the distribution on nodes in $\calL^*$ from the sampling procedure in Definition~\ref{def:thecoupling}.
  Let $\Upsilon_\sigma = \{ \rho\in \calT \mid \sigma  \in \Omega^{\Aone(\rho)\cup\Lambda}\}$.
  The probability that the sampling procedure reaches any node $\rho\in \calL^*$ with $\sigma \in \Omega^{\Aone(\rho)\cup\Lambda}$ is $P_{1,\rho}$,
  so the probability that it reaches $\Upsilon_\sigma$ is $\sum_{\rho \in \Upsilon_\sigma} P_{1,\rho}$,
  which is what we want to bound.

From the definition of $L$,
note that $\ell = 
C_0 \lceil \log(n/\eps)\rceil$ for a sufficiently large constant~$C_0$.
Suppose that the formula $\Phi$ 
 satisfies the condition in Lemma~\ref{lem:large3tree} (which it does, w.h.p.).
 Any node $\rho \in \Upsilon_\sigma$ satisfies $|\VI(\rho)| > L$, so
  by Lemma~\ref{lem:large3tree}, 
   there  is a set $T \subseteq \failedclauses(\rho)$ containing $c^*$ such that  
    $T\in \calD(G_{\Phi})$,  $|T|  = \ell$ and $\abs{T\cap \Cl_\bad} \le  {\abs{T}}/{3}$.  
This implies that
  \begin{align}\label{eqn:TandC-good}
    \abs{T\cap C^\partialVal_\good} \ge \frac{2\abs{T}}{3}.
  \end{align}
  
  Let $\zzeta$ be the set of all size-$\ell$ sets $T\in \calD(G_\Phi)$ containing~$c^*$ such that \eqref{eqn:TandC-good} holds.
  Then the probability that the sampling procedure reaches $\Upsilon_\sigma$ is at most 
  \begin{align*}
    \Pr_{\rho \sim\pi_\sigma}(\exists T\in W \text{ such that } T\subseteq \failedclauses(\rho)).
  \end{align*}
  By a union bound we have
  \begin{align*}
    \sum_{\rho \in \Upsilon_\sigma} P_{1,\rho} \leq \sum_{T \in \zzeta} \Pr_{\rho \sim\pi_\sigma}(T\subseteq \failedclauses(\rho)),
  \end{align*}
  so to finish we will show
  \begin{align}\label{eq:typetwogoal}
    \sum_{T\in \zzeta} \Pr_{\rho \sim \pi_\sigma}(T\subseteq \failedclauses(\rho)) \leq \left( k\HighD \right)^{-8\ell}.
  \end{align}
  
  Consider a set $T\in \zzeta$ and any $\rhois$ such that $\rho$   is in the support of $\pi_\sigma$
    and $T\subseteq \failedclauses(\rho)$.     By the definition of $\pi_\sigma$, 
 $ \sigma  \in \Omega^{\Aone(\rho)\cup\Lambda}$.   
 By \ref{P:failempty},
every $c\in T$ has $\reason(c)$ non-empty.

  Let $T_\good = T\cap \Cl^\partialVal_\good$ and note that $\abs{T_\good}\ge 2\ell/3$.
    By \ref{P:typebad}, every $c \in T_\good$ has 
    $\reason(c)\subseteq\{\disagree,\mathsf{1},\mathsf{2}\}$.

  Let $S = \{ c \in T_\good \mid \mathsf{1} \in\reason(c)\}$.
  By \ref{P:typetwoclauses}, every $c \in S$ has $\marked(c) \subseteq \Vset$ and $\Aone$ does not  satisfy $c$. Since $\sigma$ is $\ell$-correct, $\abs{S}<
  \ceil{\ell/2}$ so (by integrality) $|S| < \ell/2$.
  Thus, $\abs{T_\good\setminus S}\ge (2/3-1/2)\ell = \ell/6$.

  We claim that for any $c\in T_\good\setminus S$, $\disagree\in\reason(c)$.
  We have already seen that $\reason(c)$ is non-empty and that it is contained in 
  $\{\disagree, \mathsf{2}\}$.     
  If $\mathsf{2}\in\reason(c)$, then by \ref{P:typetwoclauses},
  $c$ is satisfied by $\Lambda \wedge \Aone$ but not $\Lambda \wedge \Atwo$,
  which implies that $\disagree\in\reason(c)$ by \ref{P:typedisagree}.
  
  Since $(T_\good\setminus S)\subseteq T \in \calD(G_\Phi)$,
  Definition~\ref{def:disjoint} ensures that the clauses in $T_\good\setminus S$  do not share good variables.
  Thus by \ref{P:typedisagree}, there is a set $V \subseteq \Va^\partialVal$ with 
  $|V| \geq \ell/6$ such that for every $u\in V$, $\Aone(u)\neq \Atwo(u)$.
  
  By the definition of the sampling procedure (Definition~\ref{def:thecoupling}), 
  the probability that such a disagreement occurs when the first variable~$u$ of a node~$\rho'$ is
  set is at most ${P_{1,\rho'_{\sigX,\sigXbar}}}/{P_{1,\rho'}}$ 
  for some $X\in\{\true,\false\}$.
  By \ConThree,
  this ratio is at most $1/s$.
  Also, these events are independent for the variables $u\in V$.
  Thus, 
  \begin{align*}
    \Pr_{\rho \sim \pi_\sigma}(T\subseteq \failedclauses(\rho)) \le s^{-\ell/6}.
  \end{align*}
  
   By the choice of~$C_0$, $\ell$ is at least $\log n$.
  Lemma~\ref{lem:numtrees} implies that 
  w.h.p., over the choice of~$\Phi$, 
  $\abs{\zzeta}\le (18k^2\alpha)^{4\ell}$.
  By a union bound  (using the fact that $k$ is sufficiently large and $\alpha < 2^{k/300}$),
  \begin{align*}
    \sum_{T\in \zzeta } \Pr_{\rho \sim \pi_\sigma}(T\subseteq \failedclauses(\rho)) &\leq  (18k^2\alpha)^{4\ell}s^{-\ell/6} 
    = \left( 18^4 k^8 \alpha^4 s^{-1/6} \right)^{\ell} \leq\left( k\HighD \right)^{-8\ell}. \qedhere
  \end{align*}  
\end{proof}

The following lemma is the same as Lemma~\ref{lem:typetwo}, except that we take $i=2$ rather than $i=1$.
The proof is exactly the same as that of Lemma~\ref{lem:typetwo} except that the sampling procedure 
(analogous to the one from 
Definition~\ref{def:thecoupling}) is conditioned on $\sigma\in \Omega_2^\partialVal$
and the transition from~$\rho$ to 
$\rho_{u\rightarrow X,u\rightarrow\sigma(u)}$ is with probability $ {P_{2,\rho_{u\rightarrow X,u\rightarrow\sigma(u) }}}/{P_{2,\rho}}$.

\begin{lemma}\label{lem:typetwotwo} 
  Fix $\rlower \leq \rupper$.
  Let $\ell= L/(3k^2\HighD)  $.
  \Whp over the choice of $\Phi$, the following holds.
   Let $\sigma$
  be any $\ell$-correct assignment in $  \Omega^\partialVal_2$.
  If the LP has a solution using $\rlower$ and $\rupper$ then
  $ \sum_{\rho\in \calT:\sigma \in \Omega^{\partialVal\cup\Atwo(\rho)} } P_{2,\rho} \leq \left( k\HighD \right)^{-8\ell}$.
\end{lemma}

\subsection{Soundness}

In this section we show the ``soundness'' of the LP, namely, that
whenever a solution to the LP exists, it yields a bound on   ${\abs{\Omega^\partialVal_1}}/{\abs{\Omega^\partialVal_2}}$.

\begin{lemma}\label{lem:main}
  Fix $\rlower \leq \rupper$.
  \Whp over the choice of $\Phi$, the following holds.
  If the LP has a solution $\*P$ using $\rlower$ and $\rupper$,
  then
  $  e^{-\eps/(3n)} \rlower \leq  { | \Omega^\partialVal_1|}/{|\Omega^\partialVal_2|} \leq  e^{\eps/(3n)} \rupper$.
\end{lemma}

\begin{proof}  
  By Lemma~\ref{lem:constraints2},
  the constraints in \ConTwo\ guarantee that, for any 
  $i\in \{1,2\}$ and
  $\sigma\in \Omega_i^\partialVal$,
  $$\sum_{\rho\in \calL^*:\sigma  \in \Omega^{\Ai(\rho)\cup \Lambda}} P_{i,\rho} = 1.$$
  Thus, 
  $$
  |\Omega_i^\partialVal| = \sum_{\sigma\in \Omega^\partialVal_i}1 = \sum_{\sigma\in \Omega^\partialVal_i} 
  \sum_{\rho\in \calL^*:\sigma  \in \Omega^{\Ai(\rho)\cup \Lambda}} P_{i,\rho}.
  $$
  Let $\ell = L / (3k^2\HighD)$. We start by  defining $Z_i$, $Z'_i$ and $Z''_i$ as follows
  for $i\in \{1,2\}$.
  \begin{align*}
    Z_i &= \sum_{\sigma\in \Omega_i^\partialVal}\quad 
    \sum_{\rho\in \calL: \sigma  \in \Omega^{\Ai(\rho)\cup \Lambda}} P_{i,\rho},\\
    Z'_i &= \sum_{\sigma\in \Omega_i^\partialVal,   \text{ $\sigma$    is $\ell$-wrong}}\quad 
    \sum_{\rho\in \cl^*: \sigma  \in \Omega^{\Ai(\rho)\cup \Lambda}} P_{i,\rho},\\
    Z''_i &= \sum_{\sigma\in \Omega_i^\partialVal, \text{ $\sigma$ is $\ell$-correct}}\quad 
    \sum_{\rho\in \calT: \sigma  \in\Omega^{\Ai(\rho)\cup \Lambda}} P_{i,\rho}.
  \end{align*}
  Thus $Z_i \leq \abs{\Omega_i^\partialVal}\leq Z_i + Z'_i+Z''_i$. 
  The proof consists of three parts --- as we will see soon, the statement of the lemma follows directly from Equations~\eqref{eq:goodZ},
  \eqref{eq:primesmall}, and \eqref{eq:dblprimesmall}.
  \begin{description}
    \item[Part 1] Showing 
    \begin{equation}\label{eq:goodZ}
      \rlower\leq \frac{Z_1}{Z_2} \leq \rupper. 
    \end{equation}
    \item[Part 2] Showing, for $i\in \{1,2\}$,
    \begin{equation} \label{eq:primesmall}
      \frac{Z'_i}{\abs{\Omega_i^\partialVal}} \leq \frac{1-e^{-\eps/(3n)}}{2}.
    \end{equation}
    \item[Part 3] Showing, for $i\in \{1,2\}$,
    \begin{equation} \label{eq:dblprimesmall}
      \frac{Z''_i}{\abs{\Omega_i^\partialVal}} \leq \frac{1-e^{-\eps/(3n)}}{2}.
    \end{equation}
  \end{description}  
  We now present the three parts of the proof.
  
  {\bf Part 1.\quad}  This part is straightforward. Exchanging the order of summation in the definition of~$Z_i$, we have
  \begin{equation}\label{eq:goodleaves}
    Z_i = \sum_{\rho\in \calL}\quad 
    \sum_{\sigma\in \Omega_i^\partialVal: \sigma  \in \Omega^{\Ai(\rho)\cup\Lambda}} P_{i,\rho}
    = \sum_{\rho\in \calL } P_{i,\rho} \cdot |\Omega^{\Ai(\rho)\cup \Lambda}|.
  \end{equation}
  Since $\rho\in \calL$, Lemma~\ref{lem:rho} guarantees that $r(\rho) = |\Omega^{\Aone(\rho)\cup \Lambda}|/|\Omega^{\Atwo(\rho)\cup \Lambda}|$ 
  and \ConOne\ guarantees that  
    \[\rlower\leq \frac{P_{1,\rho}\cdot \abs{\Omega^{\Aone(\rho)\cup \Lambda}}}{P_{2,\rho}\cdot \abs{\Omega^{\Atwo(\rho)\cup \Lambda}}}=\frac{P_{1,\rho}\cdot r(\rho)}{P_{2,\rho}\cdot  }\leq \rupper. \]
  Plugging in~\eqref{eq:goodleaves}, we get~\eqref{eq:goodZ}, as required.
  
  {\bf Part 2.\quad}
  Note that for any $\sigma$, $\sum_{\rho\in \cl^*: \sigma  \in \Omega^{\Ai(\rho)\cup\Lambda}} P_{i,\rho}\le 1$.
  By Lemma~\ref{lem:15}, we have that
  \[\frac{Z_i'}{\abs{\Omega_i^\partialVal}} \le 
  \frac{\abs{\{\sigma\in \Omega_i^\partialVal: \sigma \text{ is $\ell$-wrong}\}}}{\abs{\Omega_i^\partialVal}}\leq (k\HighD)^{-9\ell}\,. \]
  Since $\ell = L/(3k^2\HighD) = C_0\depthceiling$,
  we verify that
\[(k\HighD)^{-9\ell} \leq (n/\eps)^{-9C_0\log (k\HighD)}  \leq  (1-e^{-\eps/(3n)})/2\,,\]
  which implies \eqref{eq:primesmall}. This finishes Part 2.
  
  {\bf Part 3.\quad}
  By Lemmas~\ref{lem:typetwo} and~\ref{lem:typetwotwo}, 
  \Whp over the choice of~$\Phi$,
  for every $\ell$-correct $\sigma \in \Omega_i^\partialVal$, 
  $\sum_{\rho\in \calT: \sigma \in \Omega^{\Ai(\rho)\cup\Lambda}} P_{i,\rho} \leq (k\HighD)^{-8\ell}$. 
  Hence we have
  \[\frac{Z_i''}{\abs{\Omega_i^\partialVal}} \leq 
  \frac{Z_i''}{\abs{\{\sigma\in \Omega_i^\partialVal:\sigma \text{ is $\ell$-correct}\}}} \leq (k\HighD)^{-8\ell}\,. \]  
  Again, $\ell = L/(3k^2\HighD) = C_0\depthceiling$ implies~\eqref{eq:dblprimesmall}.
  This finishes Part 3.  
  
  Having finished the three parts, we now complete the proof.
  Combining~\eqref{eq:primesmall} and \eqref{eq:dblprimesmall} with the fact 
  that $Z_i\leq\abs{\Omega_i^\partialVal}\leq Z_i + Z'_i + Z''_i$, 
  we get
  \[e^{-\eps/(3n)}\leq\frac{Z_i}{\abs{\Omega^\partialVal_i}}\leq 1.\]
  Plugging in~\eqref{eq:goodZ}  we obtain
  \begin{align*}
    e^{-\eps/(3n)} \rlower &\leq \frac{ | \Omega^\partialVal_1|}{|\Omega^\partialVal_2|} \leq  e^{\eps/(3n)} \rupper. \qedhere
  \end{align*}
\end{proof}

\section{Properties of the random formula}

We still need to prove Lemmas~\ref{lem:depthoftree}, ~\ref{lem:numtrees}, and~\ref{lem:large3tree}.
All  of these lemmas depend on properties of the random formula. The main intuition behind the proofs is that individual variables within the formula might have unbounded degrees, but once we consider sets of logarithmically many vertices (that form connected sets) they behave analogously to bounded-degree sets. In Section~\ref{sec:high}, we prove some rather standard facts about the random formula $\Phi$, concerning the number of high-degree variables and its expansion properties. Then, in Section~\ref{sec:con}, we show how to bound the number of connected sets of variables in the formula $\Phi$, and formalise the intuition above by showing Lemma~\ref{lem:randomsubgraph}, from which Lemma~\ref{lem:numtrees} is derived in Section~\ref{sec:DG}. Then, in Section~\ref{sec:highcon}, we show Lemmas~\ref{lem:gvartree} and~\ref{lem:subgraphhd} which bound the number of high-degree variables in connected sets; in Section~\ref{sec:bad-vars}, we extend these  to bounds for the number of bad variables in connected sets using expansion properties of the formula. These are the key  ingredients for the proofs of Lemmas~\ref{lem:depthoftree} and~\ref{lem:large3tree}, which are given in Sections~\ref{sec:coupling-tree-property} and~\ref{sec:DG}, respectively.

Throughout this section, we use $\Pr_\Phi(\cdot)$ to denote the distribution for choosing~$\Phi$.

\subsection{Bounding the number of high-degree variables}\label{sec:high}

Recall from Section~\ref{sec:high-degree} that $\Va_0$ is the set of high-degree variables. Based on the average-degree assumption and the randomness of the formula, it is standard to show the following.
\begin{lemma}\label{lem:V013131}
  \Whp over the choice of $\Phi$, the size of $\Va_0$ is at most $n/2^{k^{10}}$.
\end{lemma}

\begin{proof}
  The degrees of the variables in $\Phi$  have the same distribution as a balls-and-bins experiment with $km$ balls and $n$ bins. Let $D_1,\hdots,D_n$ be a set of independent Poisson variables with parameter $k\alpha$, denoted $\Poi(k\alpha)$. 
  It follows by well-known facts (see, e.g., \cite[Chapter 5.4]{MU}) that the degrees of the variables in $\Phi$ have the same distribution as $\{D_1,\hdots,D_n\}$ conditioned on the event $\mathcal{E}$ that $D_1+\hdots+D_n=km$, and that $\Pr(\mathcal{E})=O(1/\sqrt{n})$. Let $U= \{i\in[n]: D_i\geq \HighD\}$, then 
  \[\Ex[\abs{U}] = n \Pr(\Poi(k\alpha) \geq \HighD)\leq ne^{-\HighD.}\leq n/2^{k^{10}+1},\]
	where the first inequality follows from $k\alpha \leq \Delta/k^2$ and using standard bounds for the tails of the Poisson distribution (see, e.g., \cite[Theorem 5.4]{MU}).   A Chernoff bound therefore yields that  $\Pr(\abs{U} \geq n/2^{k^{10}})=\exp{-\Omega(n)}$. 
  It follows that 
	\[\Pr_\Phi(\abs{\Va_0} \geq n/2^{k^{10}})= \Pr(\abs{U} \geq n/2^{k^{10}}\mid \mathcal{E}) \leq \exp{-\Omega(n)}.\qedhere\] 
\end{proof}

We will also need the following expansion properties of the random formula $\Phi$; these bound the number of clauses that can contain more than $b$ variables from a relatively small set of variables $Y$.
\begin{lemma}\label{lem:fv553dwr}
  Let $2\le b\le k$ be an integer and $t=\frac{2}{b-1}$.
  \Whp over the choice of $\Phi$, for every set of variables $Y$ such that  $2\le |Y|\leq n/2^k$, the number of clauses that contain at least $b$ variables from $Y$  is at most $t|Y|$. 
\end{lemma}
\begin{proof}
Let $y$ be an integer between $b$ and $n/2^k$.  There are $\binom{n}{y}$ ways to choose a set $Y$ of $y$ variables $Y$ and $\binom{m}{\left\lceil ty\right\rceil}$ ways to choose a set $Z$ of $\left\lceil ty\right\rceil$ clauses. 
  The probability that a clause  contains at least $b$ variables from $Y$ is bounded by $\binom{k}{b}\left( \frac{y}{n} \right)^{b}\leq (ky/n)^b$. 
  By a union bound over the choices of $Y$ and $Z$, we therefore obtain that 
  $$
    \Pr_\Phi(\exists\, Y,Z) \leq \sum_{b\le y\leq n/2^k}
    \binom{n}{y} \binom{m}{\lceil ty\rceil} \biggl(   \frac{ky}{n}  \biggr)^{b\left\lceil ty\right\rceil}
    \le \sum_{b\le y\leq n/2^k}
    \biggl(\frac{en}{y}\biggr)^{y}
    {\left(
    \frac{e\alpha n}{ty} 
    \biggl(\frac{ky}{n}\biggr)^{b }
    \right)}^{\lceil t y \rceil}.$$
    Note (by taking $k$ to be its upper bound) that the quantity taken to the $\lceil t y \rceil$ power is 
    at most
$$\frac{(b-1)e\alpha 2^k}{2} 
\biggl(\frac{k}{ 2^k}\biggr)^{b }   $$
     This is maximised at $b=2$, so it is always at most~$1$ (given that $\alpha <2^{rk}$ for some $r<1$ and that $k$ is sufficiently large).
     Thus, the quantity is maximised by removing the ceiling, so
      $$
    \Pr_\Phi(\exists\, Y,Z) \leq  \sum_{b\le y\leq n/2^k}
    {\left(
    \frac{en}{y}    
    \biggl(\frac{e\alpha n}{ty} \biggr)^{t}
    \biggl(\frac{ky}{n}\biggr)^{b t}
    \right)}^{  y  } =\sum_{b\le y\leq n/2^k}\biggl(\frac{e^{1+t}k^{bt}\alpha^{t}y}{t^tn}\biggr)^{y}=o(1). $$

  The last estimate follows from observing the inequalities   
  $\left(\frac{e^{1+t}k^{bt}\alpha^{t}y}{t^tn}\right)^y\leq 1/n$ for $2\leq y\leq \log n$,
  and $\frac{e^{1+t}k^{bt}\alpha^{t}y}{t^tn}< 1/10$ for $\log n<y\leq n/2^{k}$, which hold for all sufficiently large $n$.
\end{proof}

Applying Lemma~\ref{lem:fv553dwr} with $b=t=2$ and with $b=\lceil k/10\rceil $, $t=2/(b-1)<30/k$ gives the following two corollaries, respectively.
\begin{corollary}\label{cor:fv553dwr1}
\Whp over the choice of $\Phi$, for every set of variables $Y$ such that  $2\le |Y|\leq n/2^k$, the number of clauses that contain at least $2$ variables from $Y$  is at most $2|Y|$.
\end{corollary}

\begin{corollary}\label{cor:fv553dwr2}
\Whp over the choice of $\Phi$, for every set of variables $Y$ such that  $2\le |Y|\leq n/2^k$, the number of clauses that contain at least $k/10$ variables from $Y$  is at most $\frac{30}{k}|Y|$.
\end{corollary}

\subsection{Bounding the number of connected sets of clauses}\label{sec:con}

In this section, we bound the number of connected sets of clauses which we will use in the upcoming sections. Recall that a set $Y$ of clauses is connected if $G_\Phi[Y]$ is connected.

\begin{lemma}\label{lem:gphitree}
  For any labelled tree $T$ on a subset of clauses in $\Phi$, the probability that $T$ is a subgraph of $G_\Phi$ is at most $(k^2/n)^{\abs{V(T)}-1}$.
\end{lemma}
\begin{proof}
 For a tree $T$ on a subset of clauses in $\Phi$, we use $T\subseteq G_\Phi$ to denote that $T$ is a subgraph of $G_\Phi$.
  We prove our claim by induction on the size of $T$. 
  If $\abs{V(T)}=1$, then $T$ contains only an isolated clause, so $\Pr(T\subseteq G_\Phi) = 1$. 
  Now suppose that the claim holds for all trees with size $\abs{V(T)} - 1$. 
  Let $c$ be a leaf in $T$ and $c'$ be its neighbour in $T$. 
  Then we have
  \[
    \Pr_\Phi(T\subseteq G_\Phi) =\Pr_\Phi\big((T\setminus c)\subseteq G_\Phi\big)\, 
    \Pr_\Phi\big(\var(c)\cap \var(c') \neq \emptyset\mid (T\setminus c)\subseteq G_\Phi\big).
  \]
  For any fixed $c'$, we have that
  \begin{align*}
    \Pr(\var(c) \cap \var(c') \neq \emptyset) \leq \sum_{v\in \var(c')} \Pr(v\in \var(c)) \leq \frac{k^2}{n}.
  \end{align*}
  Note that the events $(T\setminus c)\subseteq G_\Phi$ and $\var(c)\cap \var(c') \neq \emptyset$ are independent, so
  \[
    \Pr\big(\var(c)\cap \var(c') \neq \emptyset\mid (T\setminus c)\subseteq G_\Phi\big) \leq \frac{k^2}{n}\,.
  \]
  Since $T\setminus c$ is a tree of size $\abs{V(T)} -1$, by the induction hypothesis we have that $\Pr((T\setminus c)\subseteq G_\Phi) \leq (k^2/n)^{\abs{V(T)}-2}$.  We conclude that $\Pr(T\subseteq G_\Phi) \leq (k^2/n)^{\abs{V(T)}-1}$. 
\end{proof}

\begin{lemma}\label{lem:randomsubgraph}
  \Whp over the choice of $\Phi$, for any clause $c$,  the number of connected sets of clauses in $G_\Phi$  with size $\ell \geq \log n$ containing  $c$ is at most $(9k^2\alpha)^{\ell}$.
\end{lemma}

\begin{proof}
  Let $c$ be an arbitrary clause. Let $U$ be a size-$\ell$ set of clauses   containing $c$ and let $\mathsf{T}_U$ be the set of all labelled trees on the set $U$; note that $|\mathsf{T}_U|=\ell^{\ell - 2}$. 
  For any tree $T\in \mathsf{T}_U$, by Lemma~\ref{lem:gphitree}, the probability that $T \subseteq G_\Phi$ is at most $(k^2/n)^{\ell-1}$. 
  Thus
  \[
    \Pr_\Phi(\text{$G_\Phi[U]$ is connected}) \leq \sum_{T\in \mathsf{T}_U}\Pr_\Phi(T\subseteq G_\Phi) = \ell^{\ell - 2}(k^2/n)^{\ell-1}\,.
  \]
  Let $Z_c$ be the number of connected sets of clauses with size $\ell$ containing $c$. Then,
  \begin{align*}
    \Ex_\Phi[Z_c] &= \sum_{U\subseteq \Cl; c\in U, \abs{U}=\ell}\Pr_\Phi(\text{$G_\Phi[U]$ is connected})\\
    &\leq \binom{m-1}{\ell-1}\ell^{\ell - 2}\Bigl(\frac{k^2}{n}\Bigr)^{\ell-1}
    \le \left(\frac{e(m-1)}{\ell-1}\right)^{\ell-1}\ell^{\ell - 2}\Bigl(\frac{k^2}{n}\Bigr)^{\ell-1}\\
    &\le \left( \frac{emk^2}{n}\cdot \frac{\ell^{\frac{\ell-2}{\ell-1}}}{\ell-1} \right)^{\ell-1} 
    \leq (ek^2\alpha)^{\ell-1}.
  \end{align*}
  As $\ell\ge\log n$, by Markov's inequality, 
  we obtain that $\Pr_\Phi\big(Z_c\geq (9k^2\alpha)^{\ell-1}\big)\le \left( \frac{e}{9} \right)^{\ell-1} = o(1/n)$. By a union bound over all clauses $c$ (note that there are $O(n)$ for them), we obtain  the conclusion of the lemma.
\end{proof}
  
\subsection{Bounding the number of high-degree variables in connected sets}\label{sec:highcon}
In this section, we bound the number of high-degree variables in connected sets of variables.

Recall that $\Va_0$ is the set of high-degree variables. 
For every set $S$ of variables, let $\HD(S) = \Va_0 \cap S$ be the set of high-degree variables in $S$. 
With this notation, $\Va_0=\HD(\Va)$.
For a set $Y\subseteq \Cl$ of clauses, let $\var(Y):=\cup_{c\in Y}\var(c)$.

\begin{lemma}\label{lem:gvartree}
  Let $\delta_0>0$ and $\theta_0\geq \min\{k^2 \alpha,2\}$ be constants such that $ \delta_0\theta_0\log (\theta_0/k^2\alpha) > \log \alpha+3\log k$.
  Then, \whp over the choice of $\Phi$, there do not exist sets $Y, Z$ of clauses and a set $U$  of variables such that:
  \begin{enumerate}
    \item $\abs{Y} \geq \log n$, $\abs{U}\geq \delta_0\abs{Y}$, $\abs{Z}\geq \theta_0\abs{U}$, and $Y \cap Z = \emptyset$;
    \item $G_\Phi[Y]$ is connected, $U\subseteq \var(Y)$, and every clause in $Z$ contains at least one variable from $U$.
  \end{enumerate}
\end{lemma}
\begin{proof}
  Let $\+E$ the event that there exist sets $Y,Z,U$ satisfying conditions $(1)$ and $(2)$. 
	
Call a tuple $(y,\delta,\theta)$ \emph{feasible} if $y,\delta y,\theta\delta y$ are all integers, where $y\geq \log n$,   $\delta\geq \delta_0$ and $\theta\geq \theta_0$. Fix a feasible tuple $(y, \delta, \theta)$ and three sets of indices $I_Y \in \binom{[m]}{y}, I_U \in \binom{I_Y\times [k]}{\delta y}, I_Z \in \binom{[m]\setminus I_Y}{\theta\delta y}$. Define
  \begin{align*}
  S_Y &= \{c_i \mid i \in I_Y\}\,\cr
  S_U &= \{\var(\ell_{i, j}) \mid (i, j) \in I_U\} \ \text{where $\var(\ell_{i,j})$ is the variable corresponding to the literal $\ell_{i,j}$},\cr
  S_Z &= \{c_i \mid i \in I_Z\}.
  \end{align*}  
Denote by $\+E_U$ the event that $\abs{S_U} = \delta y$, by $\+E_{Y}$ the event that $G_\Phi[S_Y]$ is connected and by $\+E_{Z}$ the event that every clause in $S_Z$ contains at least one variable from $S_U$. For any labelled tree $T$ on vertex set $S_Y$, by Lemma~\ref{lem:gphitree}, the probability that $T$ is a subgraph of $G_\Phi$ is at most $(k^2/n)^{y-1}$.  We have $y^{y-2}$ such trees, so by a union bound
  \begin{align*}
    \Pr_\Phi(\+E_Y) & \le y^{y-2} \Bigl(\frac{k^2}{n}\Bigr)^{y-1}.
  \end{align*}
  Moreover, by the independence of clauses, we have that
  \[
    \Pr_\Phi(\+E_{Z} \mid \+E_{U} \land \+E_{Y}) = \Pr_\Phi(\+E_{Z} \mid \+E_{U})  \le \Bigl(\frac{k\delta y}{n}\Bigr)^{\theta \delta y}.
  \]
Note that \[\Pr_\Phi(\+E_{U}\wedge \+E_{Y}\wedge \+E_{Z})\leq \Pr_\Phi(\+E_{Y}) \Pr_\Phi(\+E_{Z}\mid \+E_{Y}\land \+E_{U}),\] so by a union bound over the choice of the tuple $(y, \delta, \theta)$ and the sets $I_Y,I_U,I_Z$ we  have that
  \begin{align*}
    \Pr_\Phi(\+E) &\le \sum_{\text{ feasible }(y,\delta,\theta)} \sum_{I_Y \in \binom{[m]}{y}, I_U \in \binom{I_Y\times [k]}{\delta y}, I_Z \in \binom{[m]\setminus I_Y}{\theta\delta y}} \Pr_\Phi(\+E_{Y}) \Pr_\Phi(\+E_{Z}\mid \+E_{Y}\land \+E_{U}).
  \end{align*}
It follows that
  \begin{align}
    \Pr_\Phi(\+E)
    &\leq  \sum_{\text{ feasible }(y,\delta,\theta)} \binom{m}{y}\binom{ky}{\delta y}\binom{m}{\theta \delta y} y^{y-2} \Bigl(\frac{k^2}{n}\Bigr)^{y-1}  \Bigl(\frac{k\delta y}{n}\Bigr)^{\theta \delta y}\notag\\
    &\leq n \sum_{\text{ feasible }(y,\delta,\theta)}  \Big(\frac{(ek)^{2+\delta+\theta \delta}\alpha^{\theta \delta+1}}{\delta^{\delta} \theta ^{\theta \delta}}\Big)^{y}.
    \label{eq:5t4g5tg4}
  \end{align}
  Note that $\delta\theta\log (\theta/k^2\alpha)\geq \delta_0 \theta_0 \log (\theta_0/k^2\alpha)\geq \log \alpha+3\log k$ 
  and hence $\theta^{\delta\theta}\geq (k^2\alpha)^{\delta \theta}k^3\alpha$. 
  It follows that
  \[
    \frac{(ek)^{2+\delta+\theta \delta}\alpha^{\theta \delta+1}}{\delta^{\delta} \theta ^{\theta \delta}}
    \leq \frac{e^{2+\delta+\delta\theta}k^\delta}{\delta^{\delta} k^{\theta \delta+1}}=\frac{e^{2+\delta}k^\delta}{\delta^{\delta} (k/e)^{\theta \delta}k}
    \leq \frac{2e^{2+\delta}k^\delta}{k(k/e)^{2 \delta}}\leq \frac{2e^2}{k}<\frac{1}{e^6},
  \]
  where the last  few inequalities hold for sufficiently large $k$ combined with the fact that $\delta^\delta\geq 1/2$ for all $\delta>0$ and our assumption $\theta \geq \theta_0\geq 2$. 
  Plugging this estimate into \eqref{eq:5t4g5tg4} and noting that there are $O(n^3)$ feasible tuples $(y,\delta,\theta)$ and $y\geq \log n$, 
  we obtain that $\Pr_\Phi(\+E)=o(1)$, as needed.
\end{proof}

\begin{lemma}\label{lem:subgraphhd}
	\Whp over the choice of $\Phi$, every connected set $U$ of variables with size 
  at least $21600k\log n$ satisfies that $\abs{\HD(U)} \leq \frac{\abs{U}}{21600}$.
\end{lemma}
\begin{proof}
Let $\delta_0=\frac{1}{21600}$ and $\theta_0=\HighD-2(k+1)$. Note that $\delta_0\theta_0\log\frac{\theta_0}{k^2\alpha}\ge 3\log k+\log \alpha$ for all sufficiently large $k$, so \whp we have that $\Phi$ satisfies  Lemma~\ref{lem:gvartree}. Moreover, \whp we have that $\Phi$ satisfies Corollary~\ref{cor:fv553dwr1} and Lemma~\ref{lem:V013131}. We will show the conclusion of the lemma whenever $\Phi$ satisfies these properties.

For the sake of contradiction, suppose that $U$ is a connected set of variables with $|U|\geq (k/\delta_0)\log n$ such that $h >\delta_0\abs{U}$ where $h=\abs{\HD(U)}$ is the number of high-degree variables in $U$. Recall that the factor graph of $\Phi$ is a bipartite graph where one side corresponds to variables and the other to clauses (whose edges join variables to clauses in the natural way). We next show that there is a tree $T$ in the factor graph of $\Phi$ of size at most $2\abs{U}$ such that 
  \begin{enumerate}
    \item \label{it:3eda} every vertex in $T$ is either a variable in $U$ or a clause in $\Phi$, 
      all variables in $\HD(U)$ are vertices in $T$, and $T$ contains at most $\abs{U}$ clauses;
    \item \label{it:3edb} every edge in $T$ joins a variable and a clause, and for any variable $v$ and clause $c$, $(v, c)$ is an edge in $T$ only if $c$ contains $v$;
    \item \label{it:3edc} $T_L \subseteq \HD(U)$, where $T_L$ is the set of leaves of $T$;
    \item \label{it:3edd} if a clause $c\in T$ contains any variable from $\HD(U)$, then at least one of its neighbours in $T$ is a variable from $\HD(U)$.
  \end{enumerate}
  Since $H_\Phi[U]$ is connected, 
  there is a tree $T'$ of size at most $2\abs{U}$ that satisfies Items~\eqref{it:3eda} and~\eqref{it:3edb} (for example, we may take the Steiner tree with terminals $\HD(U)$ in the subgraph of the factor graph induced by $U$ and its adjacent clauses).
  We now prune $T'$ so that it satisfies Items~\eqref{it:3edd} and~\eqref{it:3edc} as well.
  For any clause $c$ in $T'$ such that $c$ contains at least one variable from $\HD(U)$ but none of its neighbours in $T'$ is from $\HD(U)$,    
  let $v$ be a variable from $\var(c) \cap \HD(U)$, and $u$ be the neighbour of $c$ on the path from $c$ to $v$. 
  Then we remove the edge $(c, u)$ from $T'$ and add the edge $(c, v)$. 
  Run this process until there is no such clause $c$. 
  Now $T'$ is a tree that satisfies Items~\eqref{it:3eda},~\eqref{it:3edb},~\eqref{it:3edd}. If $T'$ has a leaf node which is not in $\HD(U)$, remove it from $T'$. Run this process until all leaf nodes are in $\HD(U)$ and let $T$ be the remaining tree. Note that removing leaf nodes that are not in $\HD(U)$ does not affect Items~\eqref{it:3eda},~\eqref{it:3edb} and~\eqref{it:3edd}. We thus obtain a tree $T$ satisfying all of these four items.

  Let $C_T$ be the set of clauses in $T$. From Item~\eqref{it:3eda}, we have $h/k \leq \abs{C_T} \leq \abs{U}$. 
  Let $t$ be the number of clauses in $C_T$ that contain at least one variable from $\HD(U)$. 
  By Item~\eqref{it:3edd}, we have 
    \[t \leq \sum_{v\in\HD(U)} \deg_T(v).\]
  Let $D = \sum_{v\in \HD(U)\setminus T_L} \deg_T(v)$. 
  Because $T$ is a tree and $\HD(U)\subseteq T$, we obtain that
  \begin{align*}
    \abs{T_L}= 2 + \sum_{v\in T\setminus T_L} (\deg_T(v) - 2)\geq 2 + \sum_{v\in \HD(U)\setminus T_L} (\deg_T(v)-2)= 2+ D-2(h-\abs{T_L}),
  \end{align*}
  which yields that $D+\abs{T_L} \leq 2h$. 
  Thus we have $t \leq D + \abs{T_L} \leq 2h$.

  By our assumption on $U$, we have that the number of high-degree variables in $U$ satisfies $h > \delta_0\abs{U}\geq k \log n$  and therefore $\abs{C_T} \geq \log n$. 
  Each variable in $\HD(U)$ is contained in at least $\HighD$ clauses. Moreover, by Lemma~\ref{lem:V013131} we have $\abs{\HD(U)}\leq |\Va_0|\leq n/2^{k}$, so by Corollary~\ref{cor:fv553dwr1} 
  the number of clauses that contain at least $2$ variables from $\HD(U)$ is at most $2h$.
  It follows that the number of clauses that contain at least one variable from $\HD(U)$ is at least 
    $\HighD h-2h k  $. 
  Since $t \leq 2h$, at most $2h$ of these clauses appear in $T$.
  Hence, there must exist a set $Z$ of clauses of size at least $(\HighD-2(k+1))h=\theta_0 h$ such that $Z \cap C_T = \emptyset$ and each clause in $Z$ contains at least one variable from $\HD(U)$. 
	
Note that $|C_T|\geq \log n$, $\abs{\HD(U)}\geq \delta_0\abs{C_T}$ and  $\abs{Z}\geq \theta_0\abs{\HD(U)}$. 
 Moreover, $C_T$ is a connected set of clauses, $\HD(U)\subseteq \var(C_T)$ and every clause $Z$ contains at least one variable from $\HD(U)$. This contradicts that $\Phi$ satisfies Lemma~\ref{lem:gvartree}. Therefore, no such set $U$ can exist, proving the lemma.
\end{proof}

\subsection{Bounding the number of bad variables in connected sets}\label{sec:bad-vars}
In this section, we bound the number of bad variables in connected sets.  Consider the following process $\choosebad$ which we will use to work with bad components. The process, for every set $S$ of variables, defines a set of variables $\badcomp(S)$. 
\begin{enumerate}
  \item Let $\badcomp(S) = S$.
  \item While there is a clause $c$ such that
 $\aabs{\var(c) \cap \badcomp(S)}\geq k/10$ and 
  {$\var(c) \setminus \badcomp(S)  \neq \emptyset$}
 \quad $\badcomp(S) := \badcomp(S) \cup \var(c)$  
 \end{enumerate}

Note that $\+V_\bad = \badcomp(\Va_0)$, where $\Va_0$ is the set of high-degree variables. 
Recall from Section~\ref{sec:high-degree} that a bad component is a connected component of variables in $H_{\Phi,\bad}$.
The following lemma shows that the  process  $\choosebad$ can be viewed as a ``local''  process for identifying bad components.
\begin{lemma}\label{lem:induced}
For every bad component $S$, we have $S = \badcomp(\HD(S))$.
\end{lemma}
\begin{proof}
  We run the process $\choosebad$ starting from $\HD(S)$. 
  Since $\HD(S)$ consists only of high-degree variables, $\badcomp(\HD(S))\subseteq\Va_\bad$. 
  By the construction of $\badcomp(\HD(S))$ and the definitions of $\Cl_\bad$ and $H_{\Phi,\bad}$,
  $H_{\Phi,\bad}[\badcomp(\HD(S))]$ is connected.  
  Since $S$ is a connected component in $H_{\Phi,\bad}$, we obtain that $\badcomp(\HD(S)) \subseteq S$.

  Now we prove that $S\subseteq \badcomp(\HD(S))$. 
  For the sake of contradiction, suppose that $S \setminus\badcomp(\HD(S)) \neq \emptyset$. 
  Consider the process of identifying bad variables (cf. Section~\ref{sec:high-degree}).  
	Let $i$ be the smallest number such that $V_i \cap (S\setminus \badcomp(\HD(S))) \neq \emptyset$ 
  and let $v$ be any variable from $V_i \cap (S\setminus \badcomp(\HD(S)))$. 
  Thus there exists a clause $c \in C_{i-1}$ that contains $v$ and at least $k/10$ variables in $V_{i-1}$. Let $U = \var(c) \cap V_{i-1}$. 
  By Definition~\ref{def:Hphi}, variables in $U$ are adjacent to $v$ in $H_{\Phi,\bad}$ and thus they are in $S$ (since $c$ becomes a bad clause after step $i$).
  Also by the choice of $i$, $U \cap (S\setminus \badcomp(\HD(S))) \subseteq V_{i-1} \cap (S\setminus \badcomp(\HD(S))) = \emptyset$,
  which implies that $U \subseteq \badcomp(\HD(S))$. 
  Therefore, the clause $c$ contains $v$ and at least $k/10$ variables from $\badcomp(\HD(S))$, 
  so $v$ should also be in $\badcomp(\HD(S))$ according to the process $\choosebad$, which yields a contradiction. 
  This finishes the proof.
\end{proof}

To analyse the number of variables added by the process $\choosebad$, we will use an expansion property proved by Coja-Oghlan and Frieze \cite{CF14}, adapted for our purposes. First, we show the following slightly more quantitative version of   \cite[Lemma 2.4]{CF14}.
\begin{lemma}[\protect{\cite[Lemma 2.4]{CF14}}]  \label{lem:CF14-2.4}
There exists a constant $k_0>0$ such that for all $k\ge k_0$ the following holds.  With probability $1-o(1/n)$ over the choice of the random formula $\Phi$,   for $\eps>0$ and $\lambda>4$ satisfying $\eps\le k^{-3}$ and $\eps^\lambda\le\frac{1}{e}\left( 2e \right)^{-4k}$,
  $\Phi$ has the following property.    
  
  Let $Z\subset[m]$ be any set of size $\abs{Z}\le \eps n$.
    Let $i_1,\dots,i_{\ell}\in[m]\setminus Z$ be a sequence of pairwise distinct indices.
  For $s\in \{1,\ldots, \ell\}$,   
   define $N_s:=\var(Z)\cup\bigcup_{j=1}^{s-1}\var(c_{i_j})$.
    If
    \begin{align*}
      \abs{\var(c_{i_s})\cap N_s}\ge\lambda \text{ \quad for all $s\in \{1,\ldots,\ell\}$,}
    \end{align*}
    then $\ell\le\eps n$.
\end{lemma}
\begin{proof}
The proof is almost identical to the proof in \cite{CF14}, though there $\epsilon$ is a constant (independent of $n$), whereas here we allow it to depend on $n$.  \cite[Equation (4)]{CF14} shows that the probability that $\Phi$ does not have the desired property  is bounded above by $p_n:=[(2e)^{2k}\epsilon^{\lambda/2}]^{\epsilon n}$. If $\epsilon n\geq 10\log n$, then clearly $p_n=o(1/n)$ since $(2e)^{2k}\epsilon^{\lambda/2}<1/e^{1/2}$ by the assumption $\eps^\lambda\le\frac{1}{e}\left( 2e \right)^{-4k}$. If $1\leq \epsilon n\leq 10\log n$, then $(2e)^{2k}\epsilon^{\lambda/2}<n^{-3/2}$ for all sufficiently large $n$ (using $\lambda>4$) and hence $p_n=o(1/n)$. Finally, for $\epsilon n< 1$, the lemma follows by the case $\epsilon n=1$.
\end{proof}
We will use the following corollary of Lemma~\ref{lem:CF14-2.4}; note the slightly different conclusion $\ell\leq |Z|$ in the end.
\begin{corollary}\label{cor:CF14-2.4}
W.h.p.~the random formula $\Phi$ has the following property.
    Let $Z\subset[m]$ be any set of size $\abs{Z}\le 2n/2^{k^{10}}$.
    Let $i_1,\dots,i_{\ell}\in[m]\setminus Z$ be a sequence of pairwise distinct indices.
   For $s\in \{1,\ldots, \ell\}$,   
   define $N_s:=\var(Z)\cup\bigcup_{j=1}^{s-1}\var(c_{i_j})$. 
    If
    \begin{align*}
      \abs{\var(c_{i_s})\cap N_s}\ge k/10 \text{ \quad for all $s\in \{1,\ldots,\ell\}$,}
    \end{align*}
    then $\ell\le |Z|$.
\end{corollary}
\begin{proof}
For an integer $z$ satisfying $1\leq z\leq 2n/2^{k^{10}}$, let $\+E_z$ be the event that there exists a set $Z$ with $|Z|=z$ that does not satisfy the desired property. By Lemma~\ref{lem:CF14-2.4} (applied with $\epsilon=z/n$ and $\lambda=k/10$), we have $\Pr_\Phi(\mathcal{E}_z)=o(1/n)$. Taking a union bound over the values of $z$ yields the corollary.
\end{proof}

Corollary~\ref{cor:CF14-2.4} allows us to control the number of bad variables. 

\begin{lemma}\label{lem:badclause-CF}
  \Whp over the choice of $\Phi$, $|\Va_\bad|\leq 4kn/2^{k^{10}}$.
\end{lemma}
\begin{proof}
 \Whp $\Phi$ satisfies the properties in Lemma~\ref{lem:V013131} and Corollaries~\ref{cor:fv553dwr1},~\ref{cor:CF14-2.4}.  By Lemma~\ref{lem:V013131}, we have that $|\Va_0|\leq n/2^{k^{10}}$ and hence by Corollary~\ref{cor:fv553dwr1} (applied to $Y=\Va_0$), we obtain that $|\Cl_0|\leq 2n/2^{k^{10}}$.   By Lemma~\ref{lem:CF14-2.4} (applied to $Z=\Cl_0$), we conclude that $|\Cl_\bad|\leq 4n/2^{k^{10}}$ and hence $|\Va_\bad|\leq 4kn/2^{k^{10}}$.
\end{proof}

We can in fact use Corollary~\ref{cor:CF14-2.4} to  prove the following lemma.
\begin{lemma}\label{lem:sizehd}
\Whp over the choice of $\Phi$, for any bad component $S$, $\abs{S}\leq 60\abs{\HD(S)}$. 
\end{lemma}
\begin{proof}
\Whp we have that $\Phi$ satisfies the properties in  Lemma~\ref{lem:V013131} and Corollaries~\ref{cor:fv553dwr2} and~\ref{cor:CF14-2.4}. We will show the conclusion of the lemma whenever $\Phi$ satisfies these three properties.

Let $S$ be a bad component. If $S$ contains only an isolated variable, 
it must be a high-degree variable and hence $\HD(S)=S$ (so we are finished).
Otherwise,  since a bad component 
is a connected component of variables in $H_{\Phi,\bad}$,
the definition of $H_{\Phi,\bad}$ ensures that the bad component has at least $k/10$ high-degree variables.
Note that $\abs{\HD(S)}\leq |\Va_0|\leq n/2^{k^{10}}$ by Lemma~\ref{lem:V013131}.
Applying Corollary~\ref{cor:fv553dwr2} with $Y=\HD(S)$, 
we find that there are at most  $\frac{30}{k}\abs{\HD(S)}$ clauses that contain at least $k/10$ variables from $\HD(S)$. 

Now, we run the process $\choosebad$  starting with $\HD(S)$.  
Take $Z$ to be the set of clauses that contain at least $k/10$ variables from $\HD(S)$
(so, from above, we have $|Z| \leq \frac{30}{k}\abs{\HD(S)}\leq \frac{30}{k} \frac{n}{2^{k^{10}}}$). Applying    
Corollary~\ref{cor:CF14-2.4}, 
we find that the number of clauses $c$ such that $\var(c) \subseteq \badcomp(\HD(S))$ is at most $2|Z| \leq 60\abs{\HD(S)}/k$. 
Since $S=\badcomp(\HD(S))$ by Lemma~\ref{lem:induced} and each variable in $S$ is contained in some bad clause, we have 
  \[\abs{S} \leq \biggl|\bigcup_{c\in \+C_\bad:\ \var(c)\cap S\neq\emptyset}\,\var(c)\biggr| \leq 60 \abs{\HD(S)}.\qedhere\]
\end{proof}

Next, we show that there is no large bad component.
\begin{lemma}\label{lem:sizeofbad}
  \Whp over the choice of $\Phi$, every bad component $S$ has size at most $21600 k\log n$. 
\end{lemma}
\begin{proof}
  Suppose there is a bad component $S$ with size $\abs{S} > 21600k\log n$. 
  Since $S$ is a connected component in $H_{\Phi,\bad}$, $S$ is also a connected set in $H_{\Phi}$. 
  By Lemma~\ref{lem:subgraphhd}, $\abs{\HD(S)} \leq \frac{\abs{S}}{21600}$. 
  However by Lemma~\ref{lem:sizehd}, we have $\abs{S}\leq 60\abs{\HD(S)}$, which gives a contradiction.
\end{proof}

 The following lemma shows that  every ``large'' connected set contains  few bad variables.

\begin{lemma}\label{lem:badinsubgraph}
\Whp over the choice of $\Phi$,
    for every connected set $S$ of variables with size at least $21600k\log n$, 
$
    \abs{S\cap\Va_\bad}\le {\aabs{S}}/{360}$.
\end{lemma}

\begin{proof}
\Whp we have that $\Phi$ satisfies the properties in Lemmas~\ref{lem:subgraphhd} and~\ref{lem:sizehd}. We will show the conclusion of the lemma for all such $\Phi$.

    For the sake of contradiction, 
    let $S$ be a connected set of variables with size at least $21600k\log n$
    and $|S\cap\Va_{\bad}|>  {|S|}/{360}=60\delta_0\abs{S}$,
    where $\delta_0= {1}/{21600}$.  
    Suppose that there are $t$ bad components $S_1,S_2,\ldots, S_t$ intersecting $S$. 
    Let $S' = S\cup S_1 \cup \cdots \cup S_t$ and let $b =  {\aabs{S' \setminus S}}/{\aabs{S}}$.
    Note that $S'$ is  a connected set of variables. Also, all variables in $S'\setminus S$ are bad, 
    so $\abs{S'\cap \Va_\bad} > (60 \delta_0 + b)\abs{S}$. 
    Thus, by Lemma~\ref{lem:sizehd}, we have
    \begin{align*}
    \abs{\HD(S')} &= \sum_{i=1}^t\abs{\HD(S_i)} \geq \sum_{i=1}^t\frac{\abs{S_i}}{60} = \frac{\abs{S'\cap\Va_\bad}}{60}> \left(\delta_0 + \frac{b}{60}\right)\abs{S}> \delta_0(1+b)\abs{S}=\delta_0\abs{S'},
    \end{align*}
    which contradicts Lemma~\ref{lem:subgraphhd}.
\end{proof}

\begin{lemma}\label{lem:connected-bad-size}
 \Whp over the choice of $\Phi$, for every connected set of clauses $Y$ such that $\abs{\var(Y)}\ge 21600k\log n$, it holds that
$
    \aabs{Y\cap\Cl_\bad}\le {\aabs{Y}}/{12}$.
\end{lemma}
\begin{proof}
\Whp we have that $\Phi$ satisfies the properties in Corollary~\ref{cor:fv553dwr2} and Lemmas~\ref{lem:badclause-CF} and~\ref{lem:badinsubgraph}. We will show the conclusion of the lemma for all such $\Phi$.

    Let $Y$ be a connected set of clauses such that $\abs{\var(Y)}\ge 21600k\log n$ and let  $S=\var(Y)$.
    Then $\abs{S}\le k\abs{Y}$.
    Since $Y$ is connected, so is $S$.
    Let $S_\bad=S\cap\Va_\bad$ and note that, by Lemma~\ref{lem:badclause-CF}, $|S_\bad|\leq |\Va_\bad|\leq 4kn/2^{k^{10}}$.
     By  Lemma~\ref{lem:badinsubgraph}, we also have that $\abs{S_\bad}\le\abs{S}/360$.
    Note that every $c\in Y\cap\Cl_\bad$ contains at least $k/10$ variables in $S_\bad$.
    Applying
   Corollary~\ref{cor:fv553dwr2}  (with the ``$Y$'' in  the corollary equal to $S_\bad$), 
       \begin{align*}
    \abs{Y\cap\Cl_\bad}&\le \frac{30\abs{S_\bad}}{k} \le \frac{\abs{S}}{12k} \le \frac{\abs{Y}}{12}.\qedhere
    \end{align*}
\end{proof}

\subsection{Proofs for the coupling tree}
\label{sec:coupling-tree-property}
In this section, we prove Lemma~\ref{lem:depthoftree}.
For $V\subseteq\Va$, let $\Gamma_{H_\Phi}(V)=\cup_{v\in V}\Gamma_{H_\Phi}(v)$ be the neighbourhood of $V$ in $H_\Phi$.
Let $\Gamma^+_{H_\Phi}(V)=V\cup\Gamma_{H_\Phi}(V)$ be the extended neighbourhood.

\begin{lemma}  \label{lem:Gamma-H}
  W.h.p.~over the choice of $\Phi$, every connected set of variables $V\subseteq\Va$ satisfies
$
    \aabs{\Gamma^+_{H_\Phi}(V)} \le 3k^3\alpha \max\{\aabs{V},k\log n\}$.
\end{lemma}
\begin{proof}
Let $\delta_0=1$, and $\theta_0=2k^2\alpha$. Since $ \delta_0\theta_0\log (\theta_0/k^2\alpha) > \log \alpha+3\log k$, \whp we have that $\Phi$ satisfies the property in Lemma~\ref{lem:gvartree}. We will show the conclusion of the lemma for all such $\Phi$.

  Let $V$ be a connected set of variables and $Y$ be the set of neighbours of $V$ in the factor graph, i.e.,  $Y=\{c\in\Cl\mid\var(c)\cap V\neq\emptyset\}$. Clearly $\aabs{\Gamma^+_{H_\Phi}(V)}\le k\abs{Y}$ and hence it suffices to show that $\abs{Y}\le  3k^2\alpha \max\{\abs{V},k\log n\}$. There are two cases depending on the size of $V$.
  \begin{itemize}
    \item $\abs{V}\ge k\log n$. Since $V$ is a connected set of variables,
      there exists a set $Y'\subseteq Y$ such that $\abs{V}/k\le \abs{Y'}\le \abs{V}$ and $V\cup Y'$ is connected in the factor graph of $\Phi$. Hence, $Y'$ is a connected set of clauses and $\abs{Y'}\geq \log n$. Let $Z:=Y\setminus Y'$. If $\abs{Z}\geq \theta_0\abs{V}$, then we obtain a contradiction to Lemma~\ref{lem:gvartree} (using the sets $U=V, Y',Z$).
      Thus, $\abs{Z} \le \theta_0\abs{V}$ and $\abs{Y}\le \abs{Y'}+\abs{Z}\le 3k^2\alpha\abs{V}$.
    \item Otherwise $\abs{V}< k\log n$.
   If  $
    \aabs{\Gamma^+_{H_\Phi}(V)} < \lceil k \log n \rceil$ then we are finished. Otherwise, consider
     an arbitrary connected $V'\supset V$ such that $\abs{V'}=\ceil{k\log n}$.       By the argument of the previous case, the set of neighbours of $V'$ in the factor graph, denoted $Y''$, 
      satisfies that $\abs{Y''}\le3k^2\alpha\abs{V'}\le 3k^3\alpha\log n$.
      Thus, $\abs{Y}\le\abs{Y''}\le 3k^3\alpha\log n$.
  \end{itemize}
	This completes the proof.
\end{proof}

Now we can show Lemma~\ref{lem:depthoftree}, which we restate here for convenience.
Recall that $\partialVal^*$ is from Lemma~\ref{lem:partial-assignment}.
{\renewcommand{\thetheorem}{\ref{lem:depthoftree}}
	\begin{lemma}
		\statelemdepthoftree
	\end{lemma}
	\addtocounter{theorem}{-1}
}

\begin{proof} 
  W.h.p.~we have that $\Phi$ satisfies the property in Lemma~\ref{lem:Gamma-H}. We will show the conclusion of the lemma for all such $\Phi$.
	
  Let $\partialVal$ be a prefix of $\partialVal^*$ and $\rho$ be a node in $\=T^\partialVal$. 
  We first claim that $\Vset(\rho)\subseteq\Gamma^+_{H_\Phi}(\VI(\rho))$. 
  To prove the claim, we'll consider any $u\in \Vset(\rho) \setminus \VI(\rho)$ and we will show that there is a clause~$c$
  containing~$u$ and containing a variable in~$\VI(\rho)$.
  
  We first rule out the case that $u=v^*$ by noting (via Property~\ref{P:xinVset}) that $v^*\in\VI(\rho)\cap \Vset(\rho)$.
 
So consider $u\in \Vset(\rho) \setminus \VI(\rho)$ and let  $\rho'$ be the
ancestor of~$\rho$ in the coupling tree such that $u$ is the first variable of $\rho'$.
The definition of the coupling tree guarantees that~$\rho'$ is uniquely defined and that it is a proper   ancestor of~$\rho$ --- 
the definition of ``first variable'' guarantees that $u\notin \Vset(\rho')$, but for all proper descendants~$\rho'''$ of~$\rho'$,
$u\in \Vset(\rho''')$.

  Let $\rho''$ be the child of $\rho'$ on the path to $\rho$. 
  We will show that there is a clause~$c$ containing~$u$ and containing a variable in~$\VI(\rho')$.
  The claim will then follow from the fact that $\VI(\rho)$ contains $\VI(\rho')$.
  The existence of such a clause~$c$ is immediate from the definition of ``first variable'' --- indeed $c$ is the ``first clause'' of~$\rho'$.
  Thus, we have proved the claim.

  By \ref{P:connected}, $V_I(\rho)$ is a connected set of variables.
  Thus, by Lemma~\ref{lem:Gamma-H} and the claim,
  \begin{align*}
    \aabs{\Vset(\rho)}\le \aabs{\Gamma^+_{H_\Phi}(\VI(\rho))}\le 3k^3\alpha \max\{\aabs{\VI(\rho)},k\log n\}.
  \end{align*}
  If $\rho\not\in\+T$, then $\abs{\VI(\rho)}\le L$ and the lemma holds.
  Otherwise, apply the above to the parent of $\rho$, which finishes the proof.
\end{proof}

\subsection{Proofs for \texorpdfstring{$\+D(G_\Phi)$}{D(G\_Phi)}}
\label{sec:DG}

In this section we show Lemma~\ref{lem:numtrees} and Lemma~\ref{lem:large3tree}.

\begin{lemma}  \label{lem:<=4connected}
  Let $G$ be a connected graph.
  For any connected induced subgraph $G'=(V',E')$ of $G^{\leq 4}$, 
  there exists a connected induced subgraph of $G$ with size at most $4\abs{V'}$ containing all vertices in $V'$.
\end{lemma}
\begin{proof}
  We do an induction on $\ell=\abs{V'}$. 
  If $\ell = 1$ the claim holds since $G'$ is also an induced subgraph of $G$. 
  If $\ell > 1$, assume that the claim holds for all induced subgraphs of $G^{\leq 4}$ with at most $\ell - 1$ vertices. 
  Let $v$ be a vertex of $G'$ such that $G'[V'\setminus\{v\}]$ is connected in $G^{\le 4}$.
  Thus, by the induction hypothesis, 
  there exists a connected induced subgraph $G''=(V'',E'')$ of $G$ such that $(V'\setminus \{v\})\subseteq V''$ and $\abs{V''} \leq 4(\ell - 1)$. 
  Since $G'$ is connected in $G^{\leq 4}$, there exists a vertex $u \in (V'\setminus \{v\})$ such that $\dist_G(u,v)\leq 4$. 
  Let $U = V'' \cup \{\text{vertices on the path from $u$ to $v$ in $G$}\}$. 
  Then the induced subgraph in $G$ whose vertex set is $U$ is connected and $\abs{U} \leq 4\ell$. 
  Thus the claim holds for $G'$.  
\end{proof}

\begin{corollary}\label{cor:g3tog}
  Let $G$ be a connected graph and $v\in V(G)$ be a vertex. 
  Let $n_{G, \ell}(v)$ denote the number of connected induced subgraphs of $G$ with size $\ell$ containing $v$. Then
  \[
  n_{G^{\leq 4}, \ell}(v) \leq 2^{\ell'} n_{G, \ell'}(v) \mbox{ where } \ell':=\min\{4\ell,\abs{V(G)}\}.
  \]
\end{corollary}

\begin{proof}
  By Lemma~\ref{lem:<=4connected}, for any connected induced subgraph $G'$ of $G^{\leq 4}$ with size $\ell$ containing $v$,
  there exists a connected induced subgraph $G''$ of $G$ such that $V(G')\subseteq V(G'')$ and $\abs{V(G'')} \leq 4\ell$. 
  In fact we can further assume that $\abs{V(G'')} = \ell'$ since otherwise we can keep adding vertices from neighbours of $G''$ into $G''$ until $\abs{V(G'')} = \ell'$. For any such $G''$ the number of size~$\ell$ subsets containing~$v$
  (corresponding to potential graphs~$G'$  which would be mapped to~$G''$ by the above construction)
  is at most $\binom{\ell'}{\ell-1} \leq 2^{\ell'}$, giving the desired upper bound.
\end{proof}

{\renewcommand{\thetheorem}{\ref{lem:numtrees}}
	\begin{lemma}
		\statelemnumtrees
	\end{lemma}
	\addtocounter{theorem}{-1}
}
\begin{proof}
Just combine Corollary~\ref{cor:g3tog} with Lemma~\ref{lem:randomsubgraph}. \end{proof}

In the remainder of this section, we will focus on showing Lemma~\ref{lem:large3tree}.
We will need the following ingredients.

\begin{lemma}\label{lem:gvaredgesize}
  For any set $Y\subseteq \+C_{\good}$ of good  clauses,
  the size of a maximum independent set in $G_{\Phi,\good}[Y]$ is at least $ {\aabs{Y}}/{(k\Delta)}$.
\end{lemma}
\begin{proof}
  Let $c$ be a clause in $Y$. Note that $c$ contains at most $k$ variables in $\Va_\good$ and each variable in $\Va_\good$ is contained in  at most $\HighD$ clauses. So the degree of $c$ in $G_{\Phi, \good}$ is at most $k(\Delta-1)$. The result follows since every $n$-vertex graph of maximum degree $d$ contains an independent set of size at least $n/(d+1)$.
\end{proof}
 
We will also use the following properties of $\failedclauses(\rho)$.

\begin{lemma}\label{lem:connected}
  If $\rho$ is a node of the coupling tree, then
  the following  properties hold.
  \begin{enumerate}
    \item $G_\Phi^{\leq 2}[\failedclauses(\rho)]$ is connected. \label{item:f-rho-connect}
    \item $|\failedclauses(\rho)| \geq |\VI(\rho)|/k$.\label{item:f-rho-size}
  \end{enumerate}
\end{lemma}
\begin{proof}
  We show Item \eqref{item:f-rho-connect} by induction on the size of $\Vset(\rho)$.
  The base case
  where $|\Vset(\rho)|=1$
   is trivial since, in this case, $\rho=\rho^*$ and $\failedclauses(\rho^*)$ is the set of clauses containing~$v^*$.
  For the  inductive step,    we 
  consider a  node~$\rho'=\rho_{\tau_1,\tau_2}$ being created as a new child of~$\rho$ by Algorithm~\ref{alg:children} and
  we consider how  clauses are added to $\failedclauses(\rho')$. 
  We show that each part of the algorithm that adds clauses to $\failedclauses(\rho')$ maintains
  the property that $G_\Phi^{\leq 2}[\failedclauses(\rho')]$ is connected. Before Line~5, this holds by the inductive hypothesis.
   \begin{itemize}
    \item First, consider the addition of clauses in Line  \ref{step:failed-disagree}.
   All clauses~$c'$ that are added by this line contain the first variable~$u$ of~$\rho$ 
   which is in the first clause~$c$ of~$\rho$   
   so to finish it suffices to show that $\failedclauses(\rho)$
   has a clause which shares a variable with~$c$. 
   Since $\var(c) \cap \VI(\rho)$ is non-empty, it suffices to show that 
   every variable in $\VI(\rho)$ is contained in a clause in $\failedclauses(\rho)$. This is true by Property~\ref{P:VIvar}.
    \item Next, consider the addition of clauses in Line~\ref{step:failed-dissatisfy}.
    It is important to note that, after the loop containing 
     Line  \ref{step:failed-disagree},
 Property~\ref{P:VIvar} has been re-established. That is, for any $u''\in \VI$     
 there is a clause $c''\in \failedclauses$ such that $u''\in \var(c'')$.
  All clauses $c'$ added to $\failedclauses$ in   Line~\ref{step:failed-dissatisfy}
  have variables in~$\VI$ so the introduction of~$c'$ leaves 
  $G_\Phi^{\leq 2}[\failedclauses ]$ connected. Moreover, the 
  subsequent addition of variables from~$\var(c')$ to $\VI$ maintains  Property~\ref{P:VIvar}. 
  \item Finally, consider the addition of clauses in Line~ \ref{step:failed-bad}.
  As in the previous case, 
  Property~\ref{P:VIvar}   guarantees that the introduction of clauses to $\failedclauses$
  leaves 
  $G_\Phi^{\leq 2}[\failedclauses ]$ connected. Moreover, the 
  subsequent addition of variables  to $\VI$ maintains  Property~\ref{P:VIvar}.   
  \end{itemize}
  
  Item \eqref{item:f-rho-size} is a direct consequence of \ref{P:VIvar}.
\end{proof}

We also need the following expansion property (which is a strengthening of Lemma~\ref{lem:fv553dwr} in the case that $b=k$).

\begin{lemma}[\protect{\cite[Lemma 2.3]{CF14}}]\label{lem:expansion}
  For all sufficiently large $k$, \whp over the choice of $\Phi$,  
  for any $Y\subseteq\Cl$ such that $\abs{Y}\le  {n}/{k^2}$,
  $\abs{\var(Y)}\ge 0.9k\abs{Y}$.
\end{lemma}

We are now ready to prove Lemma~\ref{lem:large3tree}, which we restate here.
{\renewcommand{\thetheorem}{\ref{lem:large3tree}}
	\begin{lemma}
		\statelemlargetree
	\end{lemma}
	\addtocounter{theorem}{-1}
}
\begin{proof} 
\Whp we have that $\Phi$ satisfies the properties in Lemmas~\ref{lem:connected-bad-size} and~\ref{lem:expansion}. We will show the conclusion of the lemma for all such $\Phi$. Let $\rho \in \calT$  be a node of $\=T^\partialVal$ with $\abs{V_I(\rho)}\geq L$. For a good clause $c$, let $\Gamma_\good^+(c)$ be the set consisting of $c$ and all of its neighbours in $G_{\Phi, \good}$. Recall also from Section~\ref{sec:couplingtree} that $c^*$ is a good clause (being the first clause of the root node $\rho^*$).

Let $U = \failedclauses(\rho)\setminus (\Cl_\bad \cup\Gamma^+_\good(c^*))$ and 
  $I$ be a maximum independent set of $G_{\Phi,\good}[U]$. We let $T = I \cup \{c^*\} \cup (\failedclauses(\rho) \cap \Cl_\bad)$. 
  By Lemma~\ref{lem:gvaredgesize}
  we have $\abs{I} \geq  {\aabs{U}}/(k\HighD)$. By construction $T$ contains $c^*$.

  Next, we show that $T \in \calD(G_\Phi)$. Since $\Gamma^+_\good(c^*)$ and  $U$ are disjoint, and $I$ is an independent set of $G_{\Phi,\good}[U]$, for any $c_1,c_2\in T$ we have that  $\var(c_1) \cap \var(c_2)\cap \Va_\good = \emptyset$ (note that clauses in $\Cl_\bad$ only have bad variables).   It therefore suffices to show that $T$ is connected in $G_\Phi^{\leq 4}$. 
  
  Suppose  for contradiction that $T$ is not connected in  $G_\Phi^{\leq 4}$.  
  Then there exists a partition $(S_1, S_2)$ of $T$ such that $S_1\cup S_2 = T$, 
  $S_1 \cap S_2= \emptyset$ and $\min_{c_1\in S_1, c_2\in S_2}\dist_{G_\Phi}(c_1, c_2) \geq 5$. 
  Let $S'_i = (\cup_{c\in S_i} \Gamma^+_\good(c))\cap \failedclauses(\rho)$ for $i=1,2$. 
  Then we have $\min_{c_1\in S'_1, c_2\in S'_2}\dist_{G_\Phi}(c_1,c_2) \geq 3$.

  Since $I$ is a maximum independent set of $G_{\Phi,\good}[U]$, 
  every clause $c'$ in $U$ has $\Gamma^+_\good(c') \cap I \neq \emptyset$. 
  So $U\subseteq \cup_{c\in I} \Gamma^+_\good(c)$.
  Thus,
  \begin{align*}
    S'_1 \cup S'_2 \supseteq U \cup (\Cl_\bad \cap \failedclauses(\rho)) \cup (\Gamma^+_\good(c^*) \cap \failedclauses(\rho))= \failedclauses(\rho).
  \end{align*}
  However, $G_\Phi^{\leq 2}[\failedclauses(\rho)]$ is connected by Item \eqref{item:f-rho-connect} of Lemma~\ref{lem:connected}, 
  which contradicts 
  \begin{align*}
    \min_{c_1\in S'_1, c_2\in S'_2}\dist_{G_\Phi}(c_1,c_2) \geq 3.
  \end{align*}
  Thus, we have finished showing that $T \in \calD(G_\Phi)$.  

  Now, observe the following lower bound on the size of $T$:
  \begin{align*}
    \abs{T} &= \abs{I} +\abs{\failedclauses(\rho) \cap \Cl_\bad} +1
    \geq \frac{\abs{\failedclauses(\rho)}-\abs{\failedclauses(\rho) \cap \Cl_\bad}-k\HighD}{k\HighD}+\abs{\failedclauses(\rho) \cap \Cl_\bad}+1\cr &\geq \frac{\abs{\failedclauses(\rho)}}{k\Delta}\geq\frac{\abs{\VI(\rho)}}{k^2\HighD}\ge C_0 \depthceiling,
  \end{align*}
  where in the second to last inequality we used Item \eqref{item:f-rho-size} of Lemma~\ref{lem:connected}. If $\abs{T} > C_0\depthceiling$, we make the size of $T$ exactly equal to $C_0 \depthceiling$ by removing some clauses from it. Note that any subset of $T$ satisfies Item \eqref{item:DG-disjoint} of the definition of $\calD(G_{\Phi})$ (cf. Definition~\ref{def:disjoint}). To maintain the connectedness of $T$ in $G_{\Phi}^{\leq 4}$, consider an arbitrary spanning tree of the subgraph $G_{\Phi}^{\leq 4}[T]$; remove leaf vertices of the tree from $T$ until 
  until $\abs{T} = C_0\depthceiling$.  By construction, the remaining $T$ is still connected in $G_{\Phi}^{\leq 4}$ and hence is in $\calD(G_{\Phi})$.
  
  Finally, we show that $\abs{T\cap\Cl_\bad}\le \frac{\abs{T}}{3}$. Since $T$ is connected in $G_\Phi^{\leq 4}$,
  Lemma~\ref{lem:<=4connected} implies that there exists a connected 
  induced  set~$T'$ of vertices of~$G_\Phi$ such that $T\subseteq T'$ and $\abs{T'}\le 4\abs{T}$.
  Lemma~\ref{lem:expansion} implies that $\abs{\var(T')}\ge 0.9k\abs{T'} \geq 0.9 k |T| = 0.9 k C_0 \depthceiling> 21600k\log n$.
  Thus Lemma~\ref{lem:connected-bad-size} implies that $\abs{T'\cap\Cl_\bad}\le \frac{\abs{T'}}{12}$.
  We conclude that
  \begin{align*}
    \abs{T\cap\Cl_\bad}&\le\abs{T'\cap\Cl_\bad} \le \frac{\abs{T'}}{12} \le \frac{\abs{T}}{3}. 
  \end{align*}
This completes the proof.
\end{proof}

\section{Proof of Theorem~\ref{thm:main}} 

In order to estimate $\Omega(\Phi)$,
  we use self-reducibility to calculate the marginal probability of the partial assignment $\partialVal^*$ from Lemma~\ref{lem:partial-assignment}.
 This marginal probability is  
 ${|\Omega^{\partialVal^*}|}/|\Omega|$.
By Lemma~\ref{lem:sizeofbad}, 
w.h.p.{} over the choice of~$\Phi$,
$|\Omega^{\partialVal^*}|$ can be computed in polynomial time.
This is because
$\partialVal^*$ satisfies all of the good clauses (by Lemma~\ref{lem:partial-assignment}). The
remaining clauses are bad, and Lemma~\ref{lem:sizeofbad} guarantees that all bad components   have $O(\log n)$ size,
so their satisfying assignments can be counted by brute force.

\begin{lemma}\label{lem:marginprob}
There is a deterministic algorithm that takes as input 
a $k$-CNF formula $\Phi$ on $n$ Boolean variables with $m$ clauses, a partial assignment $\partialVal$ of $\Phi$
and an accuracy parameter~$\epsilon>0$. It returns a rational value~$p$ and runs in time $poly(n, 1/\eps)$.
If $\Phi=\Phi(k,n,m)$ then, w.h.p., the guarantees of all of our lemmas apply.
In this case, as long as $\partialVal$ is  a prefix of the partial assignment $\partialVal^*$ from Lemma~\ref{lem:partial-assignment},
the output satisfies 
$e^{-\eps/n} p\leq  {|\Omega^\partialVal_1|}/{|\Omega^\partialVal_2|} \leq e^{\eps/n} p$.
\end{lemma}

\begin{proof}
Let $v^*$ be the first unassigned variable in $\partialVal^*\setminus \partialVal$. 
  By Lemma~\ref{lem:depthoftree}, the depth of $\=T^\partialVal$ is at most $3k^3\alpha L+1 = O(\log (n/\eps))$. Thus, the number of nodes of $\=T^\partialVal$ is 
  bounded by a polynomial in~$n/\epsilon$.
   
 After constructing~$\mathbb{T}^{\partialVal}$, the algorithm
 constructs the linear program from Section~\ref{sec:LP}.
 Lemma~\ref{lem:main} 
 guarantees that, 
 given bounds $\rlower \leq \rupper$, 
  if the LP has a solution $\*P$ using $\rlower$ and $\rupper$,
  then
  $  e^{-\eps/(3n)} \rlower \leq  { | \Omega^\partialVal_1|}/{|\Omega^\partialVal_2|} \leq  e^{\eps/(3n)} \rupper$. 
Thus, the task of the algorithm
is to find initial bounds $\rlower\leq \rupper$ for which there is a feasible solution 
and then (by binary search) to bring $\rlower$ and $\rupper$ closer together, to obtain a more accurate estimate.
This is done in   Algorithm~\ref{alg:estimate}.

 The algorithm 
 is based on the premise that there is a feasible solution with $\rlower = (3s-1)/(3s+1)$ and $\rupper = (3s+1)/(3s-1)$, so
 we next establish this fact.
 This follows from Lemma~\ref{lem:local-uniform}, which guarantees that
\[
\frac{1-\frac{1}{3s}}{1+\frac{1}{3s}} \leq \frac{\abs{\Omega^\partialVal_1}}{\abs{\Omega^\partialVal_2}} \leq  \frac{1+\frac{1}{3s}}{1-\frac{1}{3s}}.
\]

\begin{algorithm}[htbp]
  \caption{$\mathsf{Estimate }\abs{\Omega^\partialVal_1}/\abs{\Omega^\partialVal_2}$}
  \label{alg:estimate}
  \begin{algorithmic}[1]  
    
    \State $\plower \gets \frac{3s-1}{3s+1}$;
    \State $\pupper \gets \frac{3s+1}{3s-1}$;
    \While{$\pupper > e^{\eps/(3n)}\plower$}
      \State $\rlower \gets \plower$;
      \State $\rupper \gets (\plower+\pupper)/2$; 
      \If{the LP described in Section~\ref{sec:LP} has a feasible solution}
        \State $\pupper \gets \rupper$;
      \Else
        \State $\plower \gets \rupper$;
      \EndIf
      \EndWhile
    \State \Return $p \gets (\plower+\pupper)/2$;
  \end{algorithmic}
\end{algorithm}

We next explore the accuracy of the output.
Applying Lemma~\ref{lem:existence} and Lemma~\ref{lem:main}, 
after each step of the while loop of Algorithm~\ref{alg:estimate}, 
it holds that $e^{-\eps/(3n)}\plower \leq \abs{\Omega^\partialVal_1}/\abs{\Omega^\partialVal_2} \leq e^{\eps/(3n)}\pupper$. 
Also, $\pupper\leq e^{\eps/(3n)}\plower$ when the algorithm terminates, 
so the output $p$ satisfies $e^{-\eps/n} p\leq \frac{\abs{\Omega^\partialVal_1}}{\abs{\Omega^\partialVal_2}} \leq e^{\eps/n} p$.

Finally, we explore the running time of the algorithm.
Since each iteration of the while loop takes polynomial time, the main issue is determining how many times the while loop executes.
Note that in the beginning $\pupper - \plower \leq 4/(3s-1)$, and  just before the last iteration, 
\[\pupper-\plower \geq (e^{\eps/(3n)}-1) \plower \geq \frac{\eps(3s-1)}{3n(3s+1)}.
\]
This difference halves with each iteration.
Thus  the algorithm solves the LP at most $O(\log(n/\eps))$ times. 
Since the size of LP is bounded by a polynomial in $n/\epsilon$ and the LP can be solved in polynomial time,
the whole of the algorithm runs in polynomial time, as required. 
\end{proof}

 We now prove Theorem~\ref{thm:main}. 

{\renewcommand{\thetheorem}{\ref{thm:main}}
	\begin{theorem}
		\statethmmain
	\end{theorem}
	\addtocounter{theorem}{-1}
}

\begin{proof} 
  Let $\Phi=\Phi(k,n,m)$ be a random formula and $\Omega =\Omega(\Phi)$. 
  The algorithm first computes $\Va_\good$, $\Va_\bad$, $\Cl_\good$ and $\Cl_\bad$ in time $poly(n)$. 
  Then using Lemma~\ref{lem:marking} and Lemma~\ref{lem:partial-assignment}, it can compute $\Va_\marked$ and $\partialVal^*$ in polynomial-time.
  \Whp over the choice of $\Phi$, the guarantees of all of our lemmas apply. Let us suppose that this happens
  (otherwise, the algorithm fails and outputs an arbitrary number).

Suppose  
$\partialVal^*$  gives values to the $j$ variables  $v_1, v_2, \ldots, v_j$. By Lemma~\ref{lem:marginprob}, for any 
$i\in \{1,\ldots,j\}$ and any
prefix $\partialVal_i\colon\{v_1, v_2,\ldots, v_{i-1}\}\to\{\true,\false\}$  of $\partialVal^*$, 
it takes $poly(n,1/\epsilon)$ time to compute  a number $p_i$ satisfying
$
e^{-\eps/n} p_i\leq   {|\Omega^{\partialVal_i}_1|}/{|\Omega^{\partialVal_i}_2|} \leq e^{\eps/n} p_i$.
If  $\partialVal^*(v_i) = \true$ let $q_i = \frac{p_i}{1+p_i}$.
Otherwise, let  $q_i = \frac{1}{1+p_i}$.
 Let $\partialVal_{j+1} = \partialVal^*$.

   Suppose that $\partialVal^*(v_i) = \true$.
 Then 
 $$
 \frac{|\Omega^{\partialVal_i}|}
 {|\Omega^{\partialVal_{i+1}}|}
 = \frac{|\Omega^{\partialVal_i}_1| + |\Omega^{\partialVal_i}_2|}{|\Omega^{\partialVal_i}_1|}
 = 1 +  \frac{  |\Omega^{\partialVal_i}_2|}{|\Omega^{\partialVal_i}_1|} 
 \leq 1 + \frac{e^{\epsilon/n} }{p_i}.
 $$
 Thus, 
 $$
  \frac{|\Omega^{\partialVal_{i+1}}|}
 {|\Omega^{\partialVal_{i}}|} 
 \geq \frac{p_i}{p_i +  {e^{\epsilon/n} } }
 \geq \frac{p_i}{e^{\epsilon/n} p_i +  {e^{\epsilon/n} } }= e^{-\epsilon/n} q_i.
 $$
 
A similar calculation for the case where $\partialVal^*(v_i) = \false$ 
and a similar calculation for the lower bound give the following.
 
\[
e^{-\eps/n} q_i\leq    \frac{| {\Omega^{\partialVal_{i+1}}}|}{ {|\Omega^{\partialVal_{i}}}|}\leq e^{\eps/n} q_i,
\]
and thus
\[
e^{-\eps}\prod_{i=1}^j q_i \leq \frac{ {|\Omega^{\partialVal^*}}|}{ {|\Omega|}} \leq e^{\eps}\prod_{i=1}^j q_i\,.
\]
Since all good clauses are satisfied by $\partialVal^*$, $\Cl^{\partialVal^*}$ consists only of bad clauses. Also, by Lemma~\ref{lem:sizeofbad}, every bad component of variables has size at most $21600 k\log n$, so $\Cl^{\partialVal^*}$ can be divided into  disjoint subsets  where each subset of clauses contains 
$O(\log n)$  variables.  The algorithm can compute the number of satisfying assignments of each subset by brute force in time $poly(n)$. Then
  $ {|\Omega^{\partialVal^*}|}$ is the product of these numbers.

Combining all of the above, our algorithm outputs
$
Z= {|\Omega^{\partialVal^*}|}\prod_{i=1}^j q_i^{-1}$,
which satisfies $e^{-\eps} {|\Omega|}\leq Z \leq e^\eps {|\Omega|}$.
\end{proof}

\section*{Acknowledgement}

We thank Andrei Constantinescu and anonymous referees for various helpful comments which improved the accuracy of the paper.

 \newpage

\appendix
\section{Notation reference}\label{sec:app}
\renewcommand{\arraystretch}{1.2}
\begin{longtable}{m{0.1\textwidth}m{0.7\textwidth}m{0.25\textwidth}}
  Name & Description & Reference\\
  \multicolumn{3}{l}{\textbf{Formula related}}\\
  $\Cl$ & The clause set of $\Phi$ where $m=\abs{\Cl} = \lfloor n \alpha \rfloor$. & Section~\ref{sec:high-degree}\\
  $\Cl_\good$,$\Cl_\bad$ & Good and bad clauses, a partition of $\Cl$ &Section~\ref{sec:high-degree}\\
  $\Va$ & The variable set of $\Phi$ where $n=\abs{\Va}$. & Section~\ref{sec:high-degree}\\
  $\Va_\good$,$\Va_\bad$ & Good and bad variables, a partition of $\Va$ &Section~\ref{sec:high-degree}\\
  $\HighD$ & The high degree threshold, set to $2^{k/300}$. &Definition~\ref{def:HD}\\
  $G_\Phi$ & The dependency graph of $\Cl$, which contains a subgraph $G_{\Phi,\good}$. & Definition~\ref{def:Gphi} \\
  $H_\Phi$ & The dependency graph of $\Va$, which contains a subgraph $H_{\Phi,\bad}$. & Definition~\ref{def:Hphi} \\
    $\calD(G_\Phi)$ & A set of subsets $T\subseteq V(G_{\Phi})$ satisfying some properties& Definition~\ref{def:disjoint}\\
  \multicolumn{3}{l}{\textbf{Local lemma}}\\
  $\Omega^*$ & All assignments $\Va\to \{\true,\false\}$ & Definition~\ref{def:lll}\\
  $\mu_A$ & Uniform distribution over $A\subseteq\Omega^*$ & Definition~\ref{def:lll}\\
  $\marked(c)$ & Marked variables in clause $c\in\Cl$ & Section~\ref{sec:local-lemma}\\
  $\Va_\marked$ & All marked variables & Section~\ref{sec:local-lemma}\\
  $\Omega$ & The set of all satisfying assignments & Section~\ref{sec:local-lemma}\\
  $\Phi^{\partialVal}$ & The formula $\Phi$ simplified under $\partialVal$ & Section~\ref{sec:local-lemma}\\
  $\Cl^\partialVal$ & Remaining clauses under $\partialVal$. Similar notations include $\Va^\partialVal$, $\Cl^\partialVal_\good$. Note that $\Va^\partialVal_\bad=\Va_\bad$ and $\Cl^\partialVal_\bad=\Cl_\bad$. & Section~\ref{sec:local-lemma}\\
  $s$ & Local uniformity parameter $s:=2^{k/4}/(ek\HighD)$. & Lemma~\ref{lem:local-uniform}\\
  $\partialVal^*$ & A particular ``nice'' partial assignment. & Lemma~\ref{lem:partial-assignment}\\
  \multicolumn{3}{l}{\textbf{Coupling tree}}\\
  $L$ & Truncation depth, set to $C_0(3k^2\HighD)\lceil \log(n/\eps)\rceil$, where $C_0$ is a sufficiently large integer (independent of $\Phi$, $k$ and $n$). & Definition~\ref{def:leaf} \\
  $\+L$ & The set of leaves of the coupling tree.& Definition~\ref{def:leaf}\\
  $\+T$ & The set of truncating nodes of the coupling tree.&Definition~\ref{def:leaf}\\
  $\+L^*$ & $\+L^*=\+L\cup\+T$. & Definition~\ref{def:leaf}\\
    \end{longtable}
\bibliographystyle{plain}
\bibliography{randksat}

\begin{thebibliography}{10}

\bibitem{AM14}
Emmanuel Abbe and Andrea Montanari.
\newblock On the concentration of the number of solutions of random
  satisfiability formulas.
\newblock {\em Random Struct. Algorithms}, 45(3):362--382, 2014.

\bibitem{AC08}
Dimitris Achlioptas and Amin Coja{-}Oghlan.
\newblock Algorithmic barriers from phase transitions.
\newblock In {\em {FOCS}}, pages 793--802. {IEEE} Computer Society, 2008.

\bibitem{2SAT}
Dimitris Achlioptas, Amin Coja-Oghlan, Max Hahn-Klimroth, Joon Lee, Noela
  M\"{u}ller, Manuel Penschuck, and Guangyan Zhou.
\newblock The random {2-SAT} partition function.
\newblock {\em arXiv e-prints}, abs/2002.03690, 2020.

\bibitem{AM02}
Dimitris Achlioptas and Cristopher Moore.
\newblock The asymptotic order of the random {$k$}-{SAT} threshold.
\newblock In {\em {FOCS}}, pages 779--788. {IEEE} Computer Society, 2002.

\bibitem{AP03}
Dimitris Achlioptas and Yuval Peres.
\newblock The threshold for random {$k$}-{SAT} is $2^k(\ln 2 - o(k))$.
\newblock In {\em {STOC}}, pages 223--231. {ACM}, 2003.

\bibitem{Alon}
Noga Alon.
\newblock A parallel algorithmic version of the {L}ocal {L}emma.
\newblock {\em Random Struct. Algorithms}, 2(4):367--378, 1991.

\bibitem{BGGGS19}
Ivona Bez{\'{a}}kov{\'{a}}, Andreas Galanis, Leslie~Ann Goldberg, Heng Guo, and
  Daniel {\v{S}}tefankovi{\v{c}}.
\newblock Approximation via correlation decay when strong spatial mixing fails.
\newblock {\em {SIAM} J. Comput.}, 48(2):279--349, 2019.

\bibitem{Blanca}
Antonio Blanca, Andreas Galanis, Leslie~Ann Goldberg, Daniel
  \v{S}tefankovi\v{c}, Eric Vigoda, and Kuan Yang.
\newblock Sampling in uniqueness from the {P}otts and random-cluster models on
  random regular graphs.
\newblock {\em SIAM Journal on Discrete Mathematics}, 34(1):742--793, 2020.

\bibitem{CGH13}
Karthekeyan Chandrasekaran, Navin Goyal, and Bernhard Haeupler.
\newblock Deterministic algorithms for the {L}ov{\'{a}}sz local lemma.
\newblock {\em {SIAM} J. Comput.}, 42(6):2132--2155, 2013.

\bibitem{Coj10}
Amin Coja-Oghlan.
\newblock A better algorithm for random {$k$}-{SAT}.
\newblock {\em SIAM J. Comput.}, 39(7):2823--2864, 2010.

\bibitem{Coj17}
Amin Coja-Oghlan.
\newblock Belief propagation guided decimation fails on random formulas.
\newblock {\em J. ACM}, 63(6):Art. 49, 55, 2017.

\bibitem{CF14}
Amin Coja-Oghlan and Alan Frieze.
\newblock Analyzing {W}alksat on random formulas.
\newblock {\em SIAM J. Comput.}, 43(4):1456--1485, 2014.

\bibitem{CHH17}
Amin Coja-Oghlan, Amir Haqshenas, and Samuel Hetterich.
\newblock {\tt {W}alksat} stalls well below satisfiability.
\newblock {\em SIAM J. Discrete Math.}, 31(2):1160--1173, 2017.

\bibitem{BPRandom}
Amin Coja{-}Oghlan, No{\"{e}}la M{\"{u}}ller, and Jean~B. Ravelomanana.
\newblock Belief propagation on the random k-sat model.
\newblock {\em CoRR}, abs/2011.02303, 2020.

\bibitem{CP16}
Amin Coja-Oghlan and Konstantinos Panagiotou.
\newblock The asymptotic {$k$}-{SAT} threshold.
\newblock {\em Adv. Math.}, 288:985--1068, 2016.

\bibitem{CR13}
Amin Coja-Oghlan and Daniel Reichman.
\newblock Sharp thresholds and the partition function.
\newblock {\em Journal of Physics: Conference Series}, 473:012015, 2013.

\bibitem{CW18}
Amin Coja-Oghlan and Nick Wormald.
\newblock The number of satisfying assignments of random regular {$k$}-{SAT}
  formulas.
\newblock {\em Combin. Probab. Comput.}, 27(4):496--530, 2018.

\bibitem{DSS15}
Jian Ding, Allan Sly, and Nike Sun.
\newblock Proof of the satisfiability conjecture for large $k$.
\newblock In {\em {STOC}}, pages 59--68. {ACM}, 2015.

\bibitem{Eft}
Charilaos Efthymiou.
\newblock A simple algorithm for sampling colorings of {${G}(n,d/n)$} up to the
  {G}ibbs uniqueness threshold.
\newblock {\em SIAM Journal on Computing}, 45(6):2087--2116, 2016.

\bibitem{EHSV18}
Charilaos Efthymiou, Thomas~P. Hayes, Daniel {\v{S}}tefankovi{\v{c}}, and Eric
  Vigoda.
\newblock Sampling random colorings of sparse random graphs.
\newblock In {\em {SODA}}, pages 1759--1771. {SIAM}, 2018.

\bibitem{EL75}
Paul Erd\H{o}s and L\'aszl\'o Lov{\'{a}}sz.
\newblock Problems and results on 3-chromatic hypergraphs and some related
  questions.
\newblock {\em Infinite and finite sets, volume 10 of Colloquia Mathematica
  Societatis J\'anos Bolyai}, pages 609--628, 1975.

\bibitem{NewHeng}
Weiming Feng, Heng Guo, Yitong Yin, and Chihao Zhang.
\newblock Fast sampling and counting $k$-{SAT} solutions in the local lemma
  regime.
\newblock In {\em Proceedings of the 52nd Annual ACM SIGACT Symposium on Theory
  of Computing}, STOC 2020, pages 854--867, 2020.

\bibitem{FR99}
Ehud Friedgut.
\newblock Sharp thresholds of graph properties, and the {$k$}-sat problem.
\newblock {\em J. Amer. Math. Soc.}, 12(4):1017--1054, 1999.
\newblock With an appendix by Jean Bourgain.

\bibitem{GJL19}
Heng Guo, Mark Jerrum, and Jingcheng Liu.
\newblock Uniform sampling through the {L}ov{\'{a}}sz local lemma.
\newblock {\em J. {ACM}}, 66(3):18:1--18:31, 2019.

\bibitem{GLLZ19}
Heng Guo, Chao Liao, Pinyan Lu, and Chihao Zhang.
\newblock Counting hypergraph colorings in the local lemma regime.
\newblock {\em {SIAM} J. Comput.}, 48(4):1397--1424, 2019.

\bibitem{HSS11}
Bernhard Haeupler, Barna Saha, and Aravind Srinivasan.
\newblock New constructive aspects of the {L}ov\'{a}sz local lemma.
\newblock {\em J. ACM}, 58(6):Art. 28, 28, 2011.

\bibitem{HSZ19}
Jonathan Hermon, Allan Sly, and Yumeng Zhang.
\newblock Rapid mixing of hypergraph independent sets.
\newblock {\em Random Struct. Algorithms}, 54(4):730--767, 2019.

\bibitem{Het16}
Samuel Hetterich.
\newblock Analysing survey propagation guided decimationon random formulas.
\newblock In {\em {ICALP}}, volume~55 of {\em LIPIcs}, pages 65:1--65:12.
  Schloss Dagstuhl - Leibniz-Zentrum fuer Informatik, 2016.

\bibitem{Jain}
Vishesh Jain, Huy~Tuan Pham, and Thuy~Duong Vuong.
\newblock On the sampling {L}ov{\'{a}}sz local lemma for atomic constraint
  satisfaction problems.
\newblock {\em CoRR}, abs/2102.08342, 2021.

\bibitem{KKKS98}
Lefteris~M. Kirousis, Evangelos Kranakis, Danny Krizanc, and Yannis~C.
  Stamatiou.
\newblock Approximating the unsatisfiability threshold of random formulas.
\newblock {\em Random Structures Algorithms}, 12(3):253--269, 1998.

\bibitem{LLLM19}
Chao Liao, Jiabao Lin, Pinyan Lu, and Zhenyu Mao.
\newblock Counting independent sets and colorings on random regular bipartite
  graphs.
\newblock In {\em {APPROX-RANDOM}}, volume 145 of {\em LIPIcs}, pages
  34:1--34:12. Schloss Dagstuhl - Leibniz-Zentrum f{\"{u}}r Informatik, 2019.

\bibitem{MMZ05}
Marc M\'ezard, Thierry Mora, and Riccardo Zecchina.
\newblock Clustering of solutions in the random satisfiability problem.
\newblock {\em Phys. Rev. Lett.}, 94:197205, 2005.

\bibitem{MPZ02}
Marc M{\'e}zard, Giorgio Parisi, and Riccardo Zecchina.
\newblock Analytic and algorithmic solution of random satisfiability problems.
\newblock {\em Science}, 297(5582):812--815, 2002.

\bibitem{MU}
Michael Mitzenmacher and Eli Upfal.
\newblock {\em Probability and Computing: Randomized Algorithms and
  Probabilistic Analysis}.
\newblock Cambridge University Press, 2005.

\bibitem{Moi19}
Ankur Moitra.
\newblock Approximate counting, the {L}ov\'{a}sz local lemma, and inference in
  graphical models.
\newblock {\em J. ACM}, 66(2):Art. 10, 25, 2019.

\bibitem{MS07}
Andrea Montanari and Devavrat Shah.
\newblock Counting good truth assignments of random $k$-{SAT} formulae.
\newblock In {\em {SODA}}, pages 1255--1264. {SIAM}, 2007.

\bibitem{MT10}
Robin~A. Moser and G\'{a}bor Tardos.
\newblock A constructive proof of the general {L}ov\'{a}sz local lemma.
\newblock {\em J. ACM}, 57(2):Art. 11, 15, 2010.

\bibitem{MS13}
Elchanan Mossel and Allan Sly.
\newblock Exact thresholds for {I}sing-{G}ibbs samplers on general graphs.
\newblock {\em Ann. Probab.}, 41(1):294--328, 2013.

\bibitem{SS85}
Jeanette Schmidt-Pruzan and Eli Shamir.
\newblock Component structure in the evolution of random hypergraphs.
\newblock {\em Combinatorica}, 5(1):81--94, 1985.

\bibitem{SSZ16}
Allan Sly, Nike Sun, and Yumeng Zhang.
\newblock The number of solutions for random regular {NAE-SAT}.
\newblock In {\em {FOCS}}, pages 724--731. {IEEE} Computer Society, 2016.

\bibitem{Spe77}
Joel Spencer.
\newblock Asymptotic lower bounds for {R}amsey functions.
\newblock {\em Discrete Math.}, 20(1):69--76, 1977.

\bibitem{YZ16}
Yitong Yin and Chihao Zhang.
\newblock Sampling in {P}otts model on sparse random graphs.
\newblock In {\em {APPROX-RANDOM}}, volume~60 of {\em LIPIcs}, pages
  47:1--47:22. Schloss Dagstuhl - Leibniz-Zentrum fuer Informatik, 2016.

\end{thebibliography}
\end{document}